\documentclass[onefignum,onetabnum]{siamonline190516}

\usepackage{amsfonts,amssymb}
\usepackage{amsopn}
\usepackage{graphicx}
\usepackage{epstopdf}
\usepackage{algorithmic}
\usepackage{soul,changepage}
\ifpdf
  \DeclareGraphicsExtensions{.eps,.pdf,.png,.jpg}
\else
  \DeclareGraphicsExtensions{.eps}
\fi

\usepackage{mathrsfs,mathtools,bbm,bm,stmaryrd}
\usepackage{nicefrac}

\usepackage{enumitem}
\setlist[enumerate]{leftmargin=.5in}
\setlist[itemize]{leftmargin=.5in}

\usepackage[scaled=.80]{beramono}

\newsiamremark{remark}{Remark}
\newsiamremark{hypothesis}{Hypothesis}
\crefname{hypothesis}{Hypothesis}{Hypotheses}
\newsiamthm{claim}{Claim}

\headers{The Modulo Radon Transform: Theory, Algorithms and Applications}{M.~Beckmann, A.~Bhandari, and F.~Krahmer}

\title{The Modulo Radon Transform: Theory, Algorithms and Applications\thanks{Our preliminary, exploratory ideas have been presented in engineering conference proceedings \cite{Bhandari2020, Beckmann2020}.
\funding{This work was funded by the UK Research and Innovation council's FLF program \emph{Sensing Beyond Barriers} (MRC Fellowship award no.~MR/S034897/1).
}}}

\author{Matthias Beckmann\thanks{Center for Industrial Mathematics, University of Bremen, Germany and Department of Mathematics, University of Hamburg, Germany 
  (\email{matthias.beckmann@uni-hamburg.de}). Between March--December 2020, Dr.~Beckmann was a visiting researcher at the Department of Electrical and Electronic Engineering, Imperial College London, UK.}
\and Ayush Bhandari\thanks{Department of Electrical and Electronic Engineering, Imperial College London, UK 
  (\email{ayush@alum.mit.edu}).}
\and Felix Krahmer\thanks{Department of Mathematics, Technical University of Munich, Germany 
  (\email{felix.krahmer@tum.de}).}}

\ifpdf
\hypersetup{
  pdftitle={The Modulo Radon Transform},
  pdfauthor={M. Beckmann, A. Bhandari, and F. Krahmer}
}
\fi

\def\N			{\mathbb N}
\def\Z			{\mathbb Z}
\def\R			{\mathbb R}

\def\Sphere		{\mathbb{S}^1}
\def\Radon		{\mathcal R}

\def\Fourier	{\mathcal F}
\def\Int		{\mathcal I}
\def\Mod		{\mathscr M}
\def\Modulo		{\mathscr R}
\def\Sum		{\mathsf S}
\def\ind		{\mathbbm 1}
\def\Cont		{\mathscr C}
\def\Lebesgue	{\mathrm L}
\def\Sobolev	{\mathrm H}

\def\Band		{\mathcal B}
\def\PW			{\mathsf{PW}}
\def\Bernstein	{\mathscr B}
\def\d			{\mathrm d}
\def\e			{\mathrm e}

\def\T			{\mathrm T}
\def\FBP		{\mathrm{FBP}}
\def\bfx		{\mathbf{x}}

\def\bftheta	{\boldsymbol{\theta}}

\newcommand     {\etal}{{et al.\ }}
\newcommand     {\sqb}[1]{[ #1 ]}

\DeclareMathOperator{\sinc}{sinc}
\DeclareMathOperator{\supp}{supp}

\begin{document}

\maketitle

\begin{abstract}
Recently, experiments have been reported where researchers were able to perform high dynamic range (HDR) tomography in a heuristic fashion, by fusing multiple tomographic projections. This approach to HDR tomography has been inspired by HDR photography and inherits the same disadvantages.
Taking a computational imaging approach to the HDR tomography problem, we here suggest a new model based on the Modulo Radon Transform (MRT), which we rigorously introduce and analyze. By harnessing a joint design between hardware and algorithms, we present a single-shot HDR tomography approach, which to our knowledge, is the only approach that is backed by mathematical guarantees. 
On the hardware front, instead of recording the Radon Transform projections that my potentially saturate, we propose to measure modulo values of the same. This ensures that the HDR measurements are folded into a lower dynamic range. On the algorithmic front, our recovery algorithms reconstruct the HDR images from folded measurements. Beyond mathematical aspects such as injectivity and inversion of the MRT for different scenarios including band-limited and approximately compactly supported images, we also provide a first proof-of-concept demonstration. To do so, we implement MRT by experimentally folding tomographic measurements available as an open source data set using our custom designed modulo hardware. Our reconstruction clearly shows the advantages of our approach for experimental data. In this way, our MRT based solution paves a path for HDR acquisition in a number of related imaging problems.
\end{abstract}



\tableofcontents


\section{Introduction}

\subsection{A Brief History of Radon Transform}

Cutting-edge advances in mathematics, hardware systems and computational resources have culminated into consumer-grade computerized tomography (CT)---a technology that has undoubtedly revolutionized medical imaging among other applications. At the conceptual core of all CT systems is the Radon Transform that links their measurements to their unknown imaging parameters. The roots of this description trace back to the early work of Minkowski \cite{Minkowski1904}, who studied the recovery of functions from the knowledge of its line integrals over big circles on a sphere. Thereon, Funk \cite{Funk1915} extended Minkowski's setup for recovery on the sphere and Radon \cite{Radon1917} developed the first solution in the context of Euclidean spaces. In the modern science and engineering landscape, the widespread implications of Radon's work are distinctly conspicuous---the Radon Transform is where mathematicians, computer scientists, signal processing experts, radiologists and hardware designers converge. 

Beyond its instrumental role in non-invasive, macroscopic imaging of the human body, the broad umbrella of methods encompassed by Radon Transform tomography \cite{Deans2007} also serves as a key enabler for a number of scientific imaging applications. Concrete examples include non-destructive testing \cite{Holler2017}, imaging of nanoscale features \cite{Dierolf2010} and biological inference \cite{Cloutier2020}. Furthermore, the sheer insights developed in designing reconstruction algorithms have catalyzed all together new imaging techniques such as non-line-of-sight or ``looking around the corners'' imaging \cite{Velten2012,Kadambi2016}. 

In the last decades, owing to its interdisciplinary implications, the Radon Transform has been studied in several contexts. In the area of applied mathematics, Tretiak and Metz \cite{Tretiak1980} have proposed the ``Exponential Radon Transform'' which generalizes the conventional Radon Transform incorporating exponential weight functions. In the context of signal processing and sampling theory, Beylkin \cite{Beylkin1987} studied the discrete Radon Transform (DRT) in spirit similar to the well known discrete Fourier Transform. Applications with reference to geophysical explorations and remote sensing have been investigated by Chapman in \cite{Chapman1981}. These advances have necessitated the development of fast and robust reconstruction algorithms \cite{Zhu2018,Miqueles2018} tailored to specific application  areas.  

The interdisciplinary applications of the Radon Transform and the cross-fertilization of the know-how have motivated a number of advanced reconstruction algorithms that have bene\-fited the field. This has been the case especially in the last decade, and now there is an increased activity towards development of reconstruction algorithms inspired by machine learning and artificial intelligence \cite{McCann2017}. In comparison, the pace of hardware development has been relatively slow. Rather, with computational resources getting cheaper and more accessible, most works have focused on new algorithms without modifying the hardware. A concrete example is that of deep learning \cite{Sahiner2018,Wang2020}. For instance, learning algorithms can be trained on a large dataset based on the standardized hardware pipeline, and the trained network is then used to infer or reconstruct information from new measurements (tomograms). While promising, all such approaches heavily rely on the quality of the hardware data. Any loss of information in the hardware pipeline will degrade the image reconstruction quality. In particular, a fundamental bottleneck is that of the dynamic range. All physical sensors operate with a fixed dynamic range and this poses a limitation on what can be measured. Information defined by signal voltages, amplitudes or intensities that are higher than the dynamic range of the sensor results in sensor saturation and this in turn yields a permanent loss of information. Handling such an information loss entirely in the reconstruction stage requires significant redundancy, which motivates us to consider alternatives based on modified hardware in this paper.

\subsection{Towards High-Dynamic-Range Tomography}
\label{sec:approach}

\subsubsection*{Limitations of Conventional, Single-Exposure Based Radon Projections} 
When working with heterogeneous materials where the tissue density varies considerably, standard CT methods typically have significant limitations. This is because source-detector calibration is critical for optimal registration of the Radon Transform projections. In particular, for X-ray radiation, the energy is parametrized by the tube voltage, which, in turn, is fixed by design, and is hence difficult to adapt to the tissue density.
As a consequence, for a fixed tissue density or thickness, if the source energy is too high, the registered projections will saturate the sensor resulting in a permanent loss of information. On the flip side, significantly lower source energy may result in inference problems due to loss of details \cite{Berkhout2004}. In either case, when the exposure is sub-optimal, repeated imaging has to be performed \cite{Haidekker2017}. In implementing practical systems with fixed source energy, there is an inherent trade-off between (a) source energy, (b) tissue density or thickness and (c) the dynamic range of the detector. When the optimal imaging conditions are not met, the recovered images exhibit undesirable artifacts. With implementation of digital X-ray systems, sensor saturation is even more pronounced because the dynamic range of digital systems is known to be lower than that of classical film-based systems \cite{Berkhout2004}. This is attributed to the power-law characteristics of the film \cite{Bushberg2012}. 
To tackle this problem, several studies have focused efforts on optimal exposure selection \cite{Ching2014}.

\subsubsection*{Recent Advances: Multiple Exposure, High-Dynamic-Range Tomography} The restrictions imposed by the detector's fixed dynamic range and the resulting difficulties in calibrating the optimal exposure have motivated new methods for tomography. Recent efforts have catalyzed what is known as \emph{high dynamic range} (HDR) tomography where the end goal is to recover tomographic information that is orders of magnitude larger than the detector's maximum range. The core idea behind such methods is inspired by consumer-grade HDR photography \cite{Debevec1997}. In crux, multiple low dynamic range measurements at different energy/exposure levels are algorithmically combined into an HDR image. The first ideas based on fusion of a pair of dental X-ray images was proposed by Trpovski \etal \cite{Trpovski2013}. A multiple exposure approach for X-ray imaging was reported by Chen \etal in \cite{Chen2015} and by Haidekker \etal in \cite{Haidekker2017}. For a concrete HDR X-ray that is based on a cat forelimb, we refer to Fig.~6 in \cite{Haidekker2017}. In the work of Weiss \etal \cite{Weiss2017}, the authors presented a pixel-level design for an HDR X-ray imaging setup. When using multi-exposure X-ray imaging, one still needs to calibrate the source energy for each exposure. To circumvent this calibration, Li \etal discuss an automated approach in \cite{Li2018}. Beyond the case of a single projection angle, recently, Chen \etal employed exposure adaption between different scanning angles in \cite{Chen2020}. 

Clearly, the approaches reviewed so far \cite{Chen2015,Haidekker2017,Li2018,Chen2020} take X-ray imaging into a new, largely unexplored direction. That said, since the HDR X-ray imaging is based on multiple exposure HDR photography, the former also inherits the limitations of HDR photography. These limitations include, 
\begin{enumerate}
\item Qualitative Aspects.

\begin{enumerate}
\item Ghosting artifacts. For exposure fusion to work \cite{Debevec1997,Chen2015,Haidekker2017}, it is necessary that all the image frames are aligned or registered. When working with multiple exposures, any movement between the exposures will result in images to be mis-aligned, thus yielding ghosting artifacts. As pointed out in the work of Haidekker \etal in \cite{Haidekker2017}, exposure times in the range of tens of seconds are required for HDR X-ray imaging and this is likely to cause motion artifacts. 

\item Exposure calibration. HDR X-rays rely on multiple exposures that should be calibrated  \cite{Li2018,Chen2020} so that the eventual reconstruction indeed reveals HDR features. 

\item Tone mapping and sensor response. For converting several exposures into a single HDR image, the knowledge of sensor response is of significant importance as this affects the reconstruction quality. In commercial systems, this aspect may be unknown \cite{Haidekker2017} due to manufacturer's policies. 
\end{enumerate}

\item Quantitative Aspects. The methods discussed in the literature are largely based on empirical experiments. A principled approach that can yield mathematically guaranteed reconstruction for HDR X-ray imaging or shed light on the number of images that are required for desirable HDR reconstruction is not available to date.  
\end{enumerate}

\begin{figure}[!t]
\centering
\includegraphics[width = 0.75\textwidth]{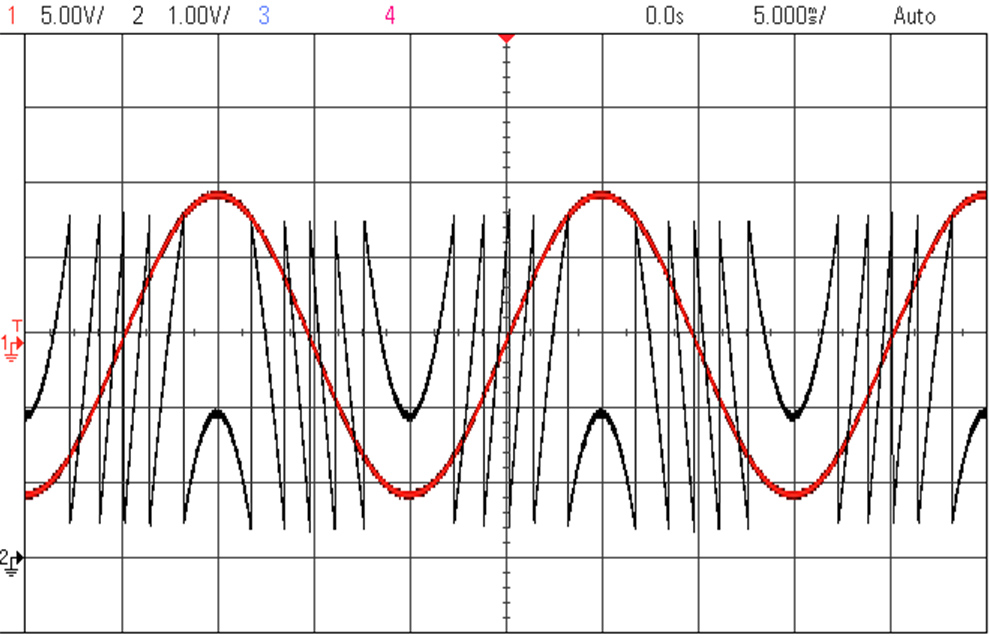}
\caption{Hardware demonstration for \emph{unlimited sampling} approach. We show an oscilloscope screenshot plotting a continuous-time function (ground truth, red) and its folded version (black). The input signal with dynamic range $20 \mathrm{V}$ peak-to-peak ($\approx 10\lambda$) is folded into a $4.025 \rm{V}$ peak-to-peak signal. Note that the ground truth (red) is plotted such that each vertical division is $5$ times larger than its corresponding value for the folded signal~(black). A live {YouTube} demonstration of this experiment is available at \href{https://youtu.be/JuZg80gUr8M}{\emph{https://youtu.be/JuZg80gUr8M}}.}
\label{fig:MTCS}
\end{figure}

\subsubsection*{Our Approach}
Motivated by the challenges and bottlenecks pivoted around HDR tomo\-graphy, in this paper, we propose a conceptually different approach which aims to overcome the aforementioned limitations. In particular, we propose a single-shot HDR imaging approach that is backed by mathematical guarantees. At the core of this new proposal is the \textbf{Modulo Radon Transform} or, in short, the MRT. 
In its implementation, the MRT is analogous to the conventional Radon Transform in that, at each angle, one computes the line integral in the Euclidean space. The unconventional aspect of the MRT is that instead of encoding pointwise Radon projections that may potentially saturate, the MRT incorporates a modulo operation that performs a reset before the saturation level is reached (see Figure~\ref{fig:MTCS} for the results of a hardware demonstration). This allows to locally capture changes of the signal values, hence circumventing sensor saturation or clipping as they appear in the multi-exposure methods analyzed in \cite{Chen2015,Weiss2017,Li2018}. At the same time, the number of resets performed is not part of the acquired data, so one encounters a different information loss.

In this paper, we discuss the resulting inverse problem and show that accurate reconstruction of the underlying tomogram is feasible under the assumption that the output of the Radon transform is band-limited or low-pass filtered.
This finite bandwidth assumption is naturally motivated by hardware implementation. Namely, for digital X-ray technology, the use of anti-aliasing filters \cite{Yaffe1997,Stark1981} is a standard choice when it comes to digital registration of a radiograph. This is to prevent aliasing of spatial information during the sampling process.
When the Radon transform is used as a forward model for imaging problems beyond tomography, the assumption often directly follows from the model. For example, in seismic imaging \cite{Foster1992,Trad2003}, it arises due to the nature of seismic pulses. 
Finally, even though it is typically not stated explicitly, the band-limitedness assumption also implicitly enters in widely used reconstruction approaches. Indeed, when it comes to commercial imaging systems, Filtered Back Projection (FBP) -- an attractive inversion scheme for the Radon Transform due to its inherent robustness to interpolation errors \cite{Stark1981} and the fact that it can be implemented efficiently -- has a low-pass filter inbuilt into its implementation.

\subsection{Scope and Organization of this Paper}

\begin{figure}[t]
\centering
\includegraphics[width = 0.95\textwidth]{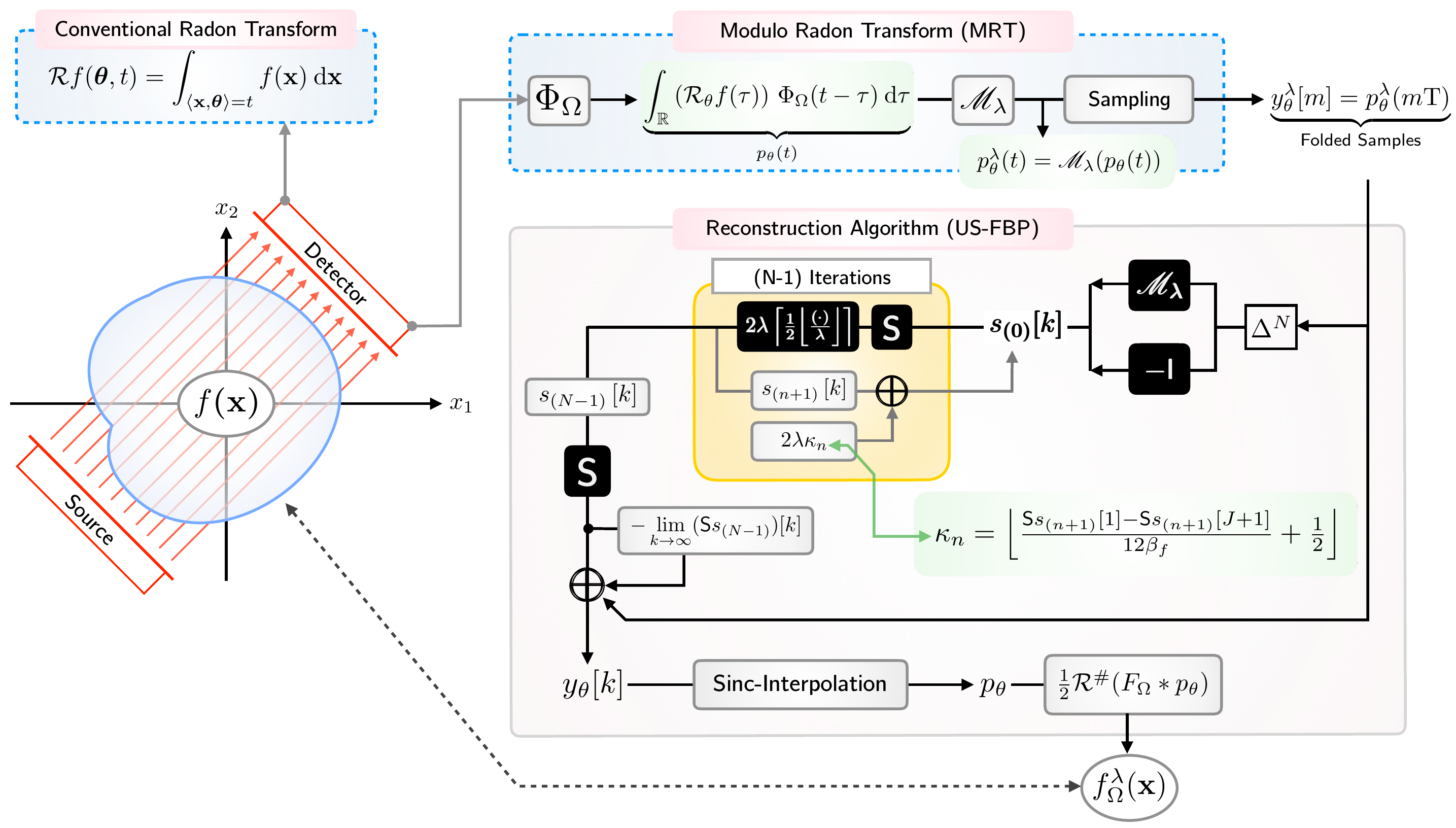}
\caption{US-FBP reconstruction scheme.}
\label{fig:Scheme}
\end{figure}

The outline of this paper is as follows.
In Section~\ref{sec:MRT} we rigorously define the Modulo Radon Transform (MRT) and study some basic properties. In particular, we prove its injectivity in Section~\ref{sec:MRT_injective} and the identifiability from semi-discrete MRT samples in Section~\ref{sec:identifiability} to prepare our discussion of the non-linear inverse problem of reconstructing a function from MRT samples in Section~\ref{sec:reconstruction}.
Our proposed reconstruction procedure for semi-discrete samples (cf.\ Section~\ref{sec:reconstruction_semi-discrete}) is summarized in Algorithm~\ref{algo:US_FBP} and schematically depicted in Figure~\ref{fig:Scheme}. It merges the Filtered Back Projection (FBP) methodology, as discussed above, with  ideas from Unlimited Sampling (US), a framework for high dynamic range sensing first proposed in~\cite{Bhandari2017, Bhandari2020a,Bhandari2020g} in the context of univariate sampling.
We prove exact recovery for band-limited target functions in Theorem~\ref{theorem:US_FBP} and provide error estimates for the case of more general classes of smooth functions in Theorem~\ref{theorem:US_FBP_error}.
Section~\ref{sec:reconstruction_finite} is then devoted to the approximate recovery from finitely many MRT projections.
The key challenge is to both incorporate the a priori information that the image is compactly supported and the implementation detail that the anti-aliasing effect of the sensors, either as a consequence of a filtering operation or due to physical limitations, leads to a smoothing. Given that the resulting functions remain close to the compactly supported Radon projections and hence exceed the modulo threshold of $\lambda$ only in a compact region, we design and analyze our method exactly for functions with compact $\lambda$-exceedance.
This insight allows for a reconstruction algorithm that is significantly simpler than Algorithm~\ref{algo:US_FBP}  and still for a rigorous recovery guarantee as made precise in Theorem~\ref{theo:UL_lambda}. As we derive for a benchmark signal model, our method can also lead to large savings in the number of samples as compared to adjusted existing approaches.
To show the applicability of our approach we finally provide numerical experiments in Section~\ref{sec:numerics}, where we in addition to the classical Shepp-Logan phantom also consider the open-source walnut dataset~\cite{Siltanen2015} and show hardware demonstrations using our custom designed prototype modulo ADC.

\section{The Modulo Radon Transform}
\label{sec:MRT}

Let $f \equiv f(\bfx)$ be a bivariate function with spatial coordinates $\bfx = (x_1,x_2) \in \R^2$.
For fixed $\lambda > 0$, we consider the centered $2\lambda$-modulo operation
\begin{equation*}
\Mod_\lambda(t) = t - 2\lambda \left\lfloor \frac{t + \lambda}{2\lambda} \right\rfloor
\quad \mbox{ for } t \in \R
\end{equation*}
and define the {\em Modulo Radon Transform} (MRT) $\Modulo^\lambda f: \Sphere \times \R \longrightarrow [-\lambda,\lambda]$ of $f$ as
\begin{equation*}
\Modulo^\lambda f(\bftheta,t) = \Mod_\lambda(\Radon f(\bftheta,t)),
\end{equation*}
where $\Radon f$ is the conventional Radon Transform of $f$ given by
\begin{equation*}
\Radon f(\bftheta,t) = \int_{\langle\bfx,\bftheta\rangle = t} f(\bfx) \: \d\bfx.
\end{equation*}
Recalling that any $\bftheta \in \Sphere$ can be uniquely expressed as $\bftheta = (\cos(\theta),\sin(\theta))$ with angle $\theta \in [0,2\pi)$, we will also use the shorthand notation
\begin{equation*}
\Radon_\theta f = \Radon f(\bftheta,\cdot)
\quad \mbox{ and } \quad
\Modulo^\lambda_\theta f = \Modulo^\lambda f(\bftheta,\cdot).
\end{equation*}

\subsection{Well-definedness and Mapping Properties}

Before discussing the inverse problem of recovering a function $f$ from samples of its Modulo Radon Transform $\Modulo^\lambda f$, we first prove some basic properties of the newly proposed operator $\Modulo^\lambda$.
We start with showing that $\Modulo^\lambda$ is well-defined on the space $\Lebesgue^1(\R^2)$ of Lebesgue integrable functions.

\begin{proposition}
\label{proposition:well_definedness}
Let $f \in \Lebesgue^1(\R^2)$ and $\lambda > 0$.
Then, for any fixed angle $\theta \in [0,2\pi)$ the univariate function $\Modulo^\lambda_\theta f$ is Lebesgue measurable and well-defined almost everywhere on $\R$. 
\end{proposition}
\begin{proof}
Let $\lambda > 0$ and $\theta \in [0,2\pi)$ be fixed.
We define the auxiliary function $H_\theta: \R^2 \longrightarrow \R^2$ as
\begin{equation*}
H_\theta(s,t) = (t \cos(\theta) - s \sin(\theta),t \sin(\theta) + s \cos(\theta))
\quad \mbox{ for } (s,t) \in \R^2.
\end{equation*}
Then, $H_\theta$ is a rotation in $\R^2$ and, thus, measure preserving.
Consequently, $f \in \Lebesgue^1(\R^2)$ implies that the mapping
\begin{equation*}
(s,t) \longmapsto f(H_\theta(s,t))
\end{equation*}
is also integrable on $\R^2$.
Hence, applying Fubini's theorem shows that the partial mapping
\begin{equation*}
s \longmapsto f(H_\theta(s,t))
\end{equation*}
is integrable on $\R$ for almost all $t \in \R$.
In particular, the mapping
\begin{equation*}
t \longmapsto \int_\R f(H(s,t,\theta)) \: \d s = \Radon_\theta f(t)
\end{equation*}
is Lebesgue measurable and well-defined for almost all $t \in \R$.
Since the centered $2\lambda$-modulo operator $\Mod_\lambda: \R \longrightarrow [-\lambda,\lambda]$ is Borel measurable and subsets of sets of Lebesgue measure zero again have Lebesgue measure zero, we can conclude that the function $\Modulo^\lambda_\theta f = \Mod_\lambda \circ \Radon_\theta f$ is also Lebesgue measurable and well-defined almost everywhere on $\R$.
\end{proof}

Note that $\Modulo^\lambda$ inherits the evenness condition
\begin{equation}
\label{eq:MRT_evenness}
\Modulo^\lambda f(-\bftheta,-t) = \Modulo^\lambda f(\bftheta,t)
\quad \mbox{ for almost all } (\bftheta,t) \in \Sphere \times \R
\end{equation}
and by the definition of $\Mod_\lambda$ we have
\begin{equation*}
\Modulo^\lambda f(\bftheta,t) \in [-\lambda,\lambda]
\quad \forall \, (\bftheta,t) \in \Sphere \times \R.
\end{equation*}

We now study some elementary mapping properties of the Modulo Radon Transform including its boundedness in the $\Lebesgue^1$-setting.

\begin{proposition}
\label{proposition:mapping_property}
For any $\lambda > 0$,
\begin{equation*}
\Modulo^\lambda_\theta: \Lebesgue^1(\R^2) \longrightarrow \Lebesgue^1(\R) \cap \Lebesgue^\infty(\R)
\end{equation*}
and
\begin{equation*}
\Modulo^\lambda: \Lebesgue^1(\R^2) \longrightarrow \Lebesgue^1(\Sphere \times \R) \cap \Lebesgue^\infty(\Sphere \times \R)
\end{equation*}
are well-defined bounded non-linear operators.
\end{proposition}

\begin{proof}
Let $f \in \Lebesgue^1(\R^2)$ and $\lambda > 0$.
According to Proposition~\ref{proposition:well_definedness}, for any $\theta \in [0,2\pi)$ the univariate function $\Modulo^\lambda_\theta f$ is Lebesgue measurable and by the definition of $\Mod_\lambda: \R \longrightarrow [-\lambda,\lambda]$ we have
\begin{equation*}
\Modulo^\lambda_\theta f(t) \in [-\lambda,\lambda]
\quad \forall \, t \in \R.
\end{equation*}
This implies $\Modulo^\lambda_\theta f \in \Lebesgue^\infty(\R)$ and $\Modulo^\lambda f \in \Lebesgue^\infty(\Sphere \times \R)$.
Moreover, recall that the definition of $\Radon_\theta f$ gives
\begin{align*}
\|\Radon_\theta f\|_{\Lebesgue^1(\R)} & = \int_\R |\Radon_\theta f(t)| \: \d t = \int_\R \bigg|\int_{\langle\bfx,\bftheta\rangle=t} f(\bfx) \: \d \bfx\bigg| \: \d t \\
& \leq \int_\R \int_\R |f(t \cos(\theta) - s \sin(\theta),t \sin(\theta) + s \cos(\theta))| \: \d s \, \d t = \|f\|_{\Lebesgue^1(\R^2)}
\end{align*}
and, thus, one has $\Radon_\theta f \in \Lebesgue^1(\R)$ for all $\theta \in [0,2\pi)$.
With this, we obtain that
\begin{align*}
\|\Modulo^\lambda_\theta f\|_{\Lebesgue^1(\R)} & = \int_{|\Radon_\theta f| > \lambda} \underbrace{|\Modulo^\lambda_\theta f(t)|}_{\leq \lambda < |\Radon_\theta f|} \: \d t + \int_{|\Radon_\theta f| \leq \lambda} \underbrace{|\Modulo^\lambda_\theta f(t)|}_{=|\Radon_\theta f(t)|} \: \d t 
\leq 
\|f\|_{\Lebesgue^1(\R^2)}
\end{align*}
so that $\Modulo^\lambda_\theta f \in \Lebesgue^1(\R)$ for any $\theta \in [0,2\pi)$ and $\Modulo^\lambda_\theta: \Lebesgue^1(\R^2) \to \Lebesgue^1(\R)$ is bounded.
Analogously, integrating also with respect to $\bftheta \in \Sphere$ gives $\Radon f \in \Lebesgue^1(\Sphere \times \R)$ implying $\Modulo^\lambda f \in \Lebesgue^1(\Sphere \times \R)$ and the boundedness of $\Modulo^\lambda: \Lebesgue^1(\R^2) \to \Lebesgue^1(\Sphere \times \R)$.
Finally, the non-linearity of $\Modulo^\lambda_\theta$ and $\Modulo^\lambda$ follows from the non-linearity of the modulo operator $\Mod_\lambda$.
\end{proof}

\subsection{Injectivity}
\label{sec:MRT_injective}

A necessary condition for the solvability of the non-linear inverse problem of reconstructing a function $f$ from its MRT data
\begin{equation*}
\bigl\{\Modulo^\lambda f(\bftheta,t) \bigm| \bftheta \in \Sphere, ~ t \in \R\bigr\}
\end{equation*}
is the injectivity of the operator $\Modulo^\lambda$ on suitable function spaces.
We now show its injectivity for two classical families of functions:
the class of compactly supported functions with continuous Radon transform, which includes most common tissue models, and
the class of band-limited functions, which models the hardware implementation as explained in Section~\ref{sec:approach}.

\smallskip

In both cases we rely on the following modulo decomposition of a univariate function $\phi \in \Cont(\R)$, see also~\cite[Proposition 1]{Bhandari2017}.

\begin{proposition}[Modulo Decomposition Property]
\label{proposition:modulo_decomposition}
Let $\lambda > 0$.
For any $\phi \in \Cont(\R)$ there exist a countable index set $J$, pairwise disjoint proper intervals $I_j \subset \R$, $j \in J$, with non-empty interior and coefficients $c_j \in \Z$, $j \in J$, such that
\begin{equation}
\label{eq:modulo_decomposition}
\phi = \Mod_\lambda(\phi) + \varepsilon_\phi
\end{equation}
with
\begin{equation*}
\varepsilon_\phi = 2\lambda \left\lfloor \frac{\phi + \lambda}{2\lambda} \right\rfloor = 2\lambda \sum_{j \in J} c_j \, \ind_{I_j}.
\end{equation*}
\end{proposition}

\begin{proof}
Let $\lambda > 0$ and $\phi \in \Cont(\R)$.
By the definition of the modulo operator $\Mod_\lambda$ we have
\begin{equation*}
\varepsilon_\phi(t) = \phi(t) - \Mod_\lambda(\phi(t)) = 2\lambda \left\lfloor \frac{\phi(t) + \lambda}{2\lambda} \right\rfloor \in 2\lambda\Z
\quad \forall \, t \in \R.
\end{equation*}
If $\phi$ is constant, the statement is trivially true.
Hence, we assume that there is $t_0 \in \R$ such that
\begin{equation*}
\phi(t_0) \in \bigl(2\lambda k, 2\lambda (k+1)\bigr)
\end{equation*}
for some $k \in \Z$.
Then, the continuity of $\phi$ implies the existence of $\delta > 0$ such that
\begin{equation*}
\phi(t) \in \bigl(2\lambda k, 2\lambda (k+1)\bigr)
\quad \forall \, t \in \bigl(t_0-\delta,t_0+\delta\bigr).
\end{equation*}
Consequently, we can construct a family $\{I_j\}_{j \in J}$ of pairwise disjoint proper intervals $I_j \subset \R$ with non-empty interior such that
\begin{equation*}
\varepsilon_\phi = 2\lambda \sum_{j \in J} c_j \, \ind_{I_j}
\end{equation*}
for suitable coefficients $c_j \in \Z$, $j \in J$.
The index set $J$ is at most countable due to the fact that each interval $I_j$, $j \in J$, contains a rational number.
\end{proof}

Based on Proposition~\ref{proposition:modulo_decomposition} we can now prove the injectivity of $\Modulo^\lambda$ for functions $f$ with continuous Radon transform $\Radon f$ vanishing at infinity in the sense that
\begin{equation*}
\Radon_\theta f \in \Cont_0(\R) = \bigl\{f \in \Cont(\R) \bigm| f(t) \to \infty \mbox{ for } |t| \to \infty\bigr\}
\end{equation*}
for all $\theta \in [0,2\pi)$.
Note that this class includes compactly supported functions with continuous Radon transform.

\begin{proposition}
\label{proposition:MRT_injectivity}
Any function $f \in \Lebesgue^1(\R^2)$ with $\Radon_\theta f \in \Cont_0(\R)$ for all $\theta \in [0,2\pi)$ is uniquely determined by its Modulo Radon Transform $\Modulo^\lambda f$.
\end{proposition}

\begin{proof}
Let $f, g \in \Lebesgue^1(\R^2)$ with $\Radon_\theta f, \Radon_\theta g \in \Cont_0(\R)$ for all $\theta \in [0,2\pi)$ and assume that
\begin{equation*}
\Modulo^\lambda f = \Modulo^\lambda g.
\end{equation*}
As $\Cont_0(\R)$ is a linear space, we have $\phi_\theta = \Radon_\theta f - \Radon_\theta g \in \Cont_0(\R)$ and modulo decomposition~\eqref{eq:modulo_decomposition} gives
\begin{equation*}
\phi_\theta = \Modulo^\lambda_\theta f + \varepsilon_{\Radon_\theta f} - \Modulo^\lambda_\theta g - \varepsilon_{\Radon_\theta g} = \varepsilon_{\Radon_\theta f} - \varepsilon_{\Radon_\theta g}.
\end{equation*}
Thus, the continuous function $\phi_\theta \in \Cont_0(\R)$ is piecewise constant and vanishing at infinity, which already implies $\phi_\theta = 0$, i.e., $\Radon_\theta f = \Radon_\theta g$ for all $\theta \in [0,2\pi)$.
Now, the well-known Fourier Slice Theorem states that
\begin{equation*}
\Fourier_1 (\Radon_\theta f)(\omega) = \Fourier_2 f(\omega\bftheta)
\quad \forall \, \omega \in \R.
\end{equation*}
This in combination with the injectivity of the Fourier transform finally gives $f = g$.
\end{proof}

We are now prepared to infer the injectivity of $\Modulo^\lambda$ for band-limited functions.
To this end, let $f$ belong to the Bernstein space $\Bernstein_\Omega^1(\R^2)$ for some bandwidth $\Omega > 0$, i.e.,
\begin{equation*}
f \in \Bernstein_\Omega^1(\R^2)
\quad \iff \quad
f \in \Lebesgue^1(\R^2) \cap \Band_\Omega(\R^2),
\end{equation*}
where $\Band_\Omega(\R^2)$ denotes the space of $\Omega$-band-limited functions, i.e.,
\begin{equation*}
f \in \Band_\Omega(\R^2)
\quad \iff \quad
\Fourier_2 f = \ind_{B_\Omega(0)} \, \Fourier_2 f.
\end{equation*}

\begin{corollary}[Injectivity of $\Modulo^\lambda$ on $\Bernstein_\Omega^1(\R^2)$]
\label{corollary:MRT_injective_bandlimited}
For any threshold $\lambda > 0$, the Modulo Radon Transform $\Modulo^\lambda$ is injective on the Bernstein space $\Bernstein_\Omega^1(\R^2)$ for any $\Omega > 0$.
\end{corollary}

\begin{proof}
Let $f \in \Bernstein_\Omega^1(\R^2)$.
Then, for any $\theta \in [0,2\pi)$ we have $\Radon_\theta f \in \Lebesgue^1(\R)$ and the Fourier Slice Theorem implies that $\Radon_\theta f \in \Bernstein_\Omega^1(\R) \subset \Cont_0(\R)$ for all $\theta \in [0,2\pi)$.
Thus, according to Proposition~\ref{proposition:MRT_injectivity} the function $f$ is uniquely determined by its Modulo Radon Transform $\Modulo^\lambda f$.
\end{proof}

As a concrete example of compactly supported functions with continuous Radon transform we finally consider the space of continuous functions $\Cont_c(\R^2)$ with compact support, i.e.,
\begin{equation*}
\Cont_c(\R^2) = \bigl\{g \in \Cont(\R^2) \bigm| \exists \, R > 0: ~ \supp(f) \subset B_R(0)\bigr\}.
\end{equation*}

\begin{corollary}[Injectivity of $\Modulo^\lambda$ on $\Cont_c(\R^2)$]
\label{corollary:MRT_injective_compact_support}
For any threshold $\lambda > 0$, the Modulo Radon Transform $\Modulo^\lambda$ is injective on the space $\Cont_c(\R^2)$ of continuous functions with compact support.
\end{corollary}

\begin{proof}
Let $f \in \Cont_c(\R^2)$ and choose $R > 0$ such that
\begin{equation*}
f(\bfx) = 0
\quad \forall \, \|\bfx\|_2 > R.
\end{equation*}
Then, we have $f \in \Lebesgue^1(\R^2)$ and for any $\theta \in [0,2\pi)$ the projection $\Radon_\theta f \in \Lebesgue^1(\R)$ is also compactly supported with
\begin{equation*}
\Radon_\theta f(t) = 0
\quad \forall \, |t| > R.
\end{equation*}
We now show the continuity of $\Radon_\theta f$.
To this end, fix $t_0 \in \R$ and consider a sequence $(t_n)_{n \in \N}$ with
\begin{equation*}
\lim_{n \to \infty} t_n = t_0.
\end{equation*}
For the sake of brevity, we introduce the notation
\begin{equation*}
f_{(t,\theta)}(s) = f(t \cos(\theta) - s \sin(\theta),t \sin(\theta) + s \cos(\theta))
\quad \mbox{ for } s \in \R
\end{equation*}
so that
\begin{equation*}
|\Radon_\theta f(t_0) - \Radon_\theta f(t_n)| = \left|\int_\R f_{(t_0,\theta)}(s) - f_{(t_n,\theta)}(s) \: \d s\right| \leq \int_{-R}^R \left|f_{(t_0,\theta)}(s) - f_{(t_n,\theta)}(s)\right| \: \d s.
\end{equation*}
Because $f \in \Cont_c(\R^2)$ is continuous and compactly supported, it is also uniformly continuous and we obtain
\begin{equation*}
\max_{|s| \leq R} \left|f_{(t_0,\theta)}(s) - f_{(t_n,\theta)}(s)\right| \xrightarrow{n \to \infty} 0.
\end{equation*}
This implies that
\begin{equation*}
|\Radon_\theta f(t_0) - \Radon_\theta f(t_n)| \leq 2R \max_{|s| \leq R} \left|f_{(t_0,\theta)}(s) - f_{(t_n,\theta)}(s)\right| \xrightarrow{n \to \infty} 0.
\end{equation*}
Thus, for any $\theta \in [0,2\pi)$ we have $\Radon_\theta f \in \Cont_c(\R) \subset \Cont_0(\R)$ and according to Proposition~\ref{proposition:MRT_injectivity} the function $f$ is uniquely determined by its Modulo Radon Transform $\Modulo^\lambda f$.
\end{proof}

In particular, Corollaries~\ref{corollary:MRT_injective_bandlimited} and~\ref{corollary:MRT_injective_compact_support} imply that the MRT is invertible as a mapping
\begin{equation*}
\Modulo^\lambda: \Bernstein_\Omega^1(\R^2) \longrightarrow \Modulo^\lambda(\Bernstein_\Omega^1(\R^2))
\end{equation*}
or
\begin{equation*}
\Modulo^\lambda: \Cont_c(\R^2) \longrightarrow \Modulo^\lambda(\Cont_c(\R^2)).
\end{equation*}

\subsection{Semi-discrete Identifiability}
\label{sec:identifiability}

In Corollary~\ref{corollary:MRT_injective_bandlimited} we have shown that any Bernstein function $f \in \Bernstein_\Omega^1(\R^2)$ is uniquely defined by its Modulo Radon Transform $\Modulo^\lambda f$ when its values are available on the full continuous domain $\Sphere \times \R$. At the same time, in the fully discrete scenario where both input variables are subsampled, identifiability is not even available when no modulo operation is applied. In this section we show that in the intermediate, semi-discrete scenario where only one variable is subsampled, i.e.,
\begin{equation*}
\bigl\{\Modulo^\lambda f(\bftheta,m\T) \bigm| \bftheta \in \Sphere, ~ m \in \Z \bigr\},
\end{equation*}
injectivity is still possible provided that the sampling rate $\T > 0$ is sufficiently small.

\begin{theorem}[Semi-discrete injectivity condition for MRT]
\label{theorem:identifiability}
Any function $f \in \Bernstein_\Omega^1(\R^2)$ is determined by its MRT samples $\bigl\{\Modulo^\lambda_\theta f(m\T) \bigm| \theta \in [0,\pi), ~ m \in \Z\bigr\}$ if the sampling rate $\T > 0$ satisfies the oversampling condition
\begin{equation}
\label{eq:optimal_oversampling_rate}
\T < \frac{\pi}{\Omega}.
\end{equation}
\end{theorem}

The proof of Theorem~\ref{theorem:identifiability} is based on the following identifiability theorem for functions from the Paley-Wiener space
\begin{equation*}
\PW_\Omega = \bigl\{\phi \in \Lebesgue^2(\R) \bigm| \supp(\Fourier_1 \phi) \subseteq [-\Omega,\Omega]\bigr\},
\end{equation*}
which is proven in~\cite{Bhandari2019}.

\begin{lemma}[{\cite[Lemma 1]{Bhandari2019}}]
\label{lemma:identifiability}
Let $\varepsilon > 0$ and $J \subset \Z$ be a finite set.
Then, any $\phi \in \PW_\Omega$ is uniquely characterized by its samples on the grid $\T_\varepsilon \cdot (\Z \setminus J)$ with $0 < \T_\varepsilon \leq \frac{\pi}{\Omega+\varepsilon}$.
\end{lemma}

With Lemma~\ref{lemma:identifiability} we are now prepared to prove the injectivity theorem.

\begin{proof}[Proof of Theorem~\ref{theorem:identifiability}]
Let $f \in \Bernstein_\Omega^1(\R^2)$ and $T < \frac{\pi}{\Omega}$.
Then, there exists $\varepsilon > 0$ such that $\T \leq \frac{\pi}{\Omega + \varepsilon}$.
Now, let $f,g \in \Bernstein_\Omega^1(\R^2)$ have the same MRT samples, i.e.,
\begin{equation*}
\Modulo^\lambda_\theta f(m \T) = \Modulo^\lambda_\theta g(m \T) \quad \forall \, \theta \in [0,\pi), ~ m \in \Z.
\end{equation*}
Recall that the Fourier Slice Theorem implies that $\Radon_\theta f, \Radon_\theta g \in \Bernstein_\Omega^1(\R) \subset \PW_\Omega$ for all $\theta \in [0,\pi)$ and the modulo decomposition~\eqref{eq:modulo_decomposition} shows that
\begin{equation*}
\psi(m\T) = \Radon_\theta f(m\T) - \Radon_\theta g(m\T) = \varepsilon_{\Radon_\theta f}(m\T) - \varepsilon_{\Radon_\theta g}(m\T) \in 2\lambda \Z
\quad \forall \, m \in \Z.
\end{equation*}
As $\PW_\Omega$ is a linear space, we also have $\psi = \Radon_\theta f - \Radon_\theta g \in \PW_\Omega \subset \PW_{\Omega+\varepsilon} \subset \Lebesgue^2(\R)$, and Shannon's sampling theorem implies
\begin{equation*}
\int_\R |\psi(t)|^2 \: \d t = \T \sum_{k \in \Z} |\psi(k\T)|^2 < \infty.
\end{equation*}
Since $\psi(m\T) \in 2\lambda\Z$ for all $m \in \Z$, the above series can only converge if all but finitely many $\psi(j\T)$, $j \in J$, are zero.
Thus, applying Lemma~\ref{lemma:identifiability} shows $\psi = 0$, i.e.,
\begin{equation*}
\Radon_\theta f = \Radon_\theta g \quad \forall \, \theta \in [0,\pi).
\end{equation*}
This in combination with the evenness condition~\eqref{eq:MRT_evenness} and again the Fourier Slice Theorem as well as the injectivity of the Fourier transform finally gives $f = g$.
\end{proof}


\section{Reconstruction Schemes}
\label{sec:reconstruction}

In Theorem~\ref{theorem:identifiability} we have seen that any band-limited function $f \in \Lebesgue^1(\R^2)$ is uniquely determined by its semi-discrete MRT samples
\begin{equation*}
\bigl\{\Modulo^\lambda f(\bftheta,m\T) \bigm| \bftheta \in \Sphere, ~ m \in \Z \bigr\}
\end{equation*}
provided that the sampling rate $\T > 0$ satisfies the oversampling condition~$\T < \frac{\pi}{\Omega}$.
The proof, however, is non-constructive and we now aim at providing tractable reconstruction algorithms with provable recovery guarantees.
We proceed by first devising a reconstruction strategy for the semi-discrete setting~\eqref{eq:MRT_samples_discrete}, making use of ideas from Unlimited Sampling (US), as first proposed in~\cite{Bhandari2017, Bhandari2020a,Bhandari2020g} and adapted to our setting.
The key ingredient to transfer these findings to the fully discrete case \eqref{eq:MRT_samples_fully_discrete} is the notion of a compact $\lambda$-exceedance, which allows us to deal with a finite number of samples per projection angle.

\subsection{Recovery from Semi-Discrete MRT Samples}
\label{sec:reconstruction_semi-discrete}

We start with describing a sequential reconstruction scheme for recovering $f \in \Bernstein_\Omega^1(\R^2)$ with known bandwidth $\Omega > 0$ from given semi-discrete MRT samples
\begin{equation}
\label{eq:MRT_samples_discrete}
\bigl\{\Modulo^\lambda_\theta f(m\T) \bigm| \theta \in [0,\pi), ~ m \in \Z\bigr\}.
\end{equation}
In the first step we recover $\Radon f$ by utilizing a variant of Unlimited Sampling (US), which makes use of the forward difference operator $\Delta: \R^\Z \longrightarrow \R^\Z$, $\Delta a\sqb{k} = a[k+1] - a\sqb{k}$, as well as the anti-difference operator $\Sum: \R^\Z \longrightarrow \R^\Z$,
\begin{equation*}
\Sum a\sqb{k} = \begin{cases}
\sum_{j=0}^{k-1} a[j] & \text{for } k > 0\\
0 & \text{for } k = 0\\
-\sum_{j=k}^{-1} a[j] & \text{for } k < 0.
\end{cases}
\end{equation*}
In the second step, we finally recover $f$ by applying the filtered back projection (FBP) formula
\begin{equation}
\label{eq:FBP}
f_\Omega = \frac{1}{2} \Radon^\# (F_\Omega * \Radon_\theta f),
\end{equation}
where $F_\Omega \in \PW_\Omega$ denotes a low-pass reconstruction filter of the form
\begin{equation}
\label{eq:recon_filter}
\Fourier_1 F_\Omega(S) = |S| \, W(\nicefrac{S}{\Omega})
\end{equation}
with an even window $W \in \Lebesgue^\infty(\R)$ supported in $[-1,1]$ and $\Radon^\#$ is the back projection operator
\begin{equation*}
\Radon^\# g(\bfx) = \frac{1}{2\pi} \int_{\Sphere} g(\bftheta,\bfx^\top \bftheta) \: \d\bftheta.
\end{equation*}
Our recovery scheme is summarized in Algorithm~\ref{algo:US_FBP}.
\begin{algorithm}[t]
\caption{Unlimited Sampling--Filtered Back Projection (US-FBP)}
\label{algo:US_FBP}
\begin{algorithmic}[1]
\REQUIRE samples 
$y_\theta^\lambda\sqb{k} = p_\theta^\lambda(k\T)$ for $k \in \Z$ and $\theta \in [0,\pi)$, uniform bound $2\lambda\Z \ni \beta_f \geq \|p_\theta\|_\infty$
\medskip
\STATE \textbf{choose} $N = \left\lceil\frac{\log(\lambda) - \log(\beta_f)}{\log(\T\Omega\e)}\right\rceil$, $J = 6 \frac{\beta_f}{\lambda}$
\smallskip
\STATE $s_{(0)}\sqb{k} = \bigl(\Mod_\lambda(\Delta^N y_\theta^\lambda) - \Delta^N y_\theta^\lambda\bigr)\sqb{k}$
\FOR{$n=0,\ldots,N-2$}
\STATE $s_{(n+1)}\sqb{k} = 2\lambda \left\lceil\frac{\lfloor\nicefrac{\Sum s_{(n)}\sqb{k}}{\lambda}\rfloor}{2}\right\rceil$
\smallskip
\STATE $\kappa_n = \left\lfloor\frac{\Sum s_{(n+1)}[1] - \Sum s_{(n+1)}[J+1]}{12 \beta_f} + \frac{1}{2}\right\rfloor$
\smallskip
\STATE $s_{(n+1)}\sqb{k} = s_{(n+1)}\sqb{k} + 2 \lambda \kappa_n$
\ENDFOR
\STATE $y_\theta\sqb{k} = y_\theta^\lambda\sqb{k} + (\Sum s_{(N-1)})\sqb{k} - \lim_{j \to \infty} (\Sum s_{(N-1)})[j]$
\STATE $p_\theta = \sum_{k \in \Z} y_\theta\sqb{k} \sinc\left(\frac{\pi}{\T}(\cdot - k\T)\right)$
\medskip
\ENSURE US-FBP reconstruction $f_\Omega^\lambda = \frac{1}{2} \Radon^\# (F_\Omega * p_\theta)$
\end{algorithmic}
\end{algorithm}
In the following theorem we give a sufficient condition for reconstructing $f \in \Bernstein_\Omega^1(\R^2)$ from MRT samples~\eqref{eq:MRT_samples_discrete} by applying Algorithm~\ref{algo:US_FBP} with
\begin{equation*}
p_\theta = \Radon_\theta f
\quad \mbox{ and } \quad
p_\theta^\lambda = \Modulo_\theta^\lambda f.
\end{equation*}

\begin{theorem}
\label{theorem:US_FBP}
Any $f \in \Bernstein_\Omega^1(\R^2)$ can be uniquely recovered from semi-discrete MRT samples~\eqref{eq:MRT_samples_discrete} by applying Algorithm~\ref{algo:US_FBP} if
\begin{equation*}
\T < \frac{1}{\Omega\e}
\end{equation*}
and the reconstruction filter $F_\Omega$ is chosen to be the Ram-Lak filter with window $W = \ind_{[-1,1]}$.
\end{theorem}

The proof of Theorem~\ref{theorem:US_FBP} is based on the following lemma, which is proven in~\cite{Bhandari2017}.

\begin{lemma}[{\cite[Lemma 1]{Bhandari2017}}]
\label{lemma:difference_norm}
Let $\phi \in \PW_\Omega$ and $N \in \N$.
Then, the samples $\gamma\sqb{k} = \phi(k\T)$, $k \in \Z$, satisfy
\begin{equation*}
\|\Delta^N \gamma\|_\infty \leq (\T \Omega \e)^N \, \|g\|_\infty.
\end{equation*}
\end{lemma}

We are now prepared to prove Theorem~\ref{theorem:US_FBP}.
To this end, for $k \in \Z$ we consider
\begin{equation*}
y_\theta^\lambda\sqb{k} = p_\theta^\lambda(k\T) = \Modulo_\theta^\lambda f(k\T), \qquad
y_\theta\sqb{k} = p_\theta(k\T) = \Radon_\theta f(k\T)
\end{equation*}
and set
\begin{equation*}
\varepsilon_\theta\sqb{k} = y_\theta\sqb{k} - y_\theta^\lambda\sqb{k} \in 2\lambda\Z.
\end{equation*}
Let us recall that for any sequence $a \in \R^\Z$ we have $\Sum(\Delta a) = a - a[0]$.
In particular, if $a \in \R^\Z$ is a null sequence, we obtain
\begin{equation*}
a[0] = -\lim_{k \to \infty} \Sum(\Delta a)\sqb{k}
\end{equation*}
and, thus, any $a \in c_0(\Z)$ can be uniquely recovered from $\Delta a$ via
\begin{equation*}
a = \Sum(\Delta a) - \lim_{k \to \infty} \Sum(\Delta a)\sqb{k}.
\end{equation*}

\begin{proof}[Proof of Theorem~\ref{theorem:US_FBP}]
First note that due to the choice of $\beta_f \in 2\lambda \Z$ we have
\begin{equation*}
\|\Radon_\theta f\|_\infty \leq \beta_f
\quad \forall \, \theta \in [0,\pi)
\end{equation*}
and the assumption $\T < \nicefrac{1}{\Omega\e}$ gives
\begin{equation*}
N = \left\lceil\frac{\log(\lambda) - \log(\beta_f)}{\log(\T\Omega\e)}\right\rceil \geq \frac{\log(\nicefrac{\lambda}{\beta_f})}{\log(\T \Omega \e)}
\quad \implies \quad
(\T \Omega \e)^N \leq \frac{\lambda}{\beta_f}.
\end{equation*}
As $f \in \Bernstein_\Omega^1(\R^2)$ implies $\Radon_\theta f \in \PW_\Omega$, we can apply Lemma~\ref{lemma:difference_norm} to obtain
\begin{equation}
\label{eq:ybound}
\|\Delta^N y_\theta\|_\infty \leq (\T\Omega\e)^N \, \|\Radon_\theta f\|_\infty \leq \lambda.
\end{equation}
and, consequently,
\begin{equation*}
\Delta^N y_\theta = \Mod_\lambda\bigl(\Delta^N y_\theta\bigr) = \Mod_\lambda\bigl(\Delta^N y_\theta - \Delta^N \varepsilon_\theta\bigr) = \Mod_\lambda\bigl(\Delta^N y_\theta^\lambda\bigr),
\end{equation*}
since $\Delta^N \varepsilon_\theta\sqb{k} \in 2\lambda\Z$ for all $k \in \Z$.
Hence, $\Delta^N \varepsilon_\theta$ can be computed from given MRT samples $y_\theta^\lambda$ via
\begin{equation*}
\Delta^N \varepsilon_\theta = \Mod_\lambda\bigl(\Delta^N y_\theta^\lambda\bigr) - \Delta^N y_\theta^\lambda.
\end{equation*}
We now prove by induction in $m \in \{0,\ldots,N-1\}$ that $s_{(m)} = \Delta^{N-m} \varepsilon_\theta$.
The induction seed reduces to the definition of $s_{(0)} = \Delta^N \varepsilon_\theta$.
For the induction step, we assume that for fixed $m$ we have $s_{(m)} = \Delta^{N-m} \varepsilon_\theta$.
Recall that $\varepsilon_\theta\sqb{k} \in 2\lambda\Z$ so that $\bigl(\Delta^j \varepsilon_\theta\bigr)\sqb{k} \in 2\lambda\Z$ for all $k \in \Z$ and $j \in \N$.
Consequently, applying the anti-difference operator $\Sum$ yields
\begin{equation*}
\Delta^{N-m-1} \varepsilon_\theta = \Sum \Delta^{N-m} \varepsilon_\theta + 2\lambda \kappa
\end{equation*}
with $\kappa \in \Z$.
Since $(\Sum \Delta^{N-m} \varepsilon_\theta\bigr)\sqb{k} \in 2\lambda\Z$, rounding to the nearest multiple of $2\lambda$ has no effect, i.e.,
\begin{equation*}
\Sum \Delta^{N-m} \varepsilon_\theta = 2\lambda \left\lceil\frac{\left\lfloor\nicefrac{\Sum \Delta^{N-m} \varepsilon_\theta}{\lambda}\right\rfloor}{2}\right\rceil,
\end{equation*}
and it remains to identify the ambiguity $\kappa$ as
\begin{equation*}
\kappa = \left\lfloor\frac{\bigl(\Sum^2 \Delta^{N-m} \varepsilon_\theta\bigr)[1] - \bigl(\Sum^2 \Delta^{N-m} \varepsilon_\theta\bigr)[J+1]}{12 \beta_f} + \frac{1}{2}\right\rfloor.
\end{equation*}
To this end, we again apply the anti-difference operator $\Sum$ and obtain, with some $\mu \in \Z$,
\begin{equation*}
\bigl(\Delta^{N-m-2} \varepsilon_\theta\bigr)\sqb{k} = \bigl(\Sum \Delta^{N-m-1} \varepsilon_\theta\bigr)\sqb{k} + 2\lambda \mu = \bigl(\Sum^2 \Delta^{N-m} \varepsilon_\theta\bigr)\sqb{k} + 2\lambda\kappa k + 2\lambda\mu
\end{equation*}
so that
\begin{equation*}
\zeta^{(m)}_J = \bigl(\Sum^2 \Delta^{N-m} \varepsilon_\theta\bigr)[1] - \bigl(\Sum^2 \Delta^{N-m} \varepsilon_\theta\bigr)[J+1] = \bigl(\Delta^{N-m-2} \varepsilon_\theta\bigr)[1] - \bigl(\Delta^{N-m-2} \varepsilon_\theta\bigr)[J+1] + 2\lambda\kappa J.
\end{equation*}
As $\|y_\theta^\lambda\|_\infty \leq \lambda$ gives $\|\Delta^{N-m-2} y_\theta^\lambda\|_\infty \leq 2^{N-m-2} \lambda$ and $\|\Delta^{N-m-2} y_\theta\|_\infty \leq (\T\Omega\e)^{N-m-2} \beta_f \leq \beta_f$ by applying Lemma~\ref{lemma:difference_norm} with $\T\Omega\e < 1$, we obtain, using $\Delta^{N-m-2} \varepsilon_\theta = \Delta^{N-m-2} y_\theta - \Delta^{N-m-2} y_\theta^\lambda$,
\begin{equation*}
\zeta^{(m)}_J \in \left[2\lambda\kappa J - (2\beta_f + 2^{N-m-1} \, \lambda), 2\lambda\kappa J + (2\beta_f + 2^{N-m-1} \, \lambda)\right].
\end{equation*}
The choices of $\T$ and $N$ ensure that
\begin{equation*}
2^{N-1} \leq \left(\frac{\beta_f}{\lambda}\right)^{-\frac{1}{\log(\T\Omega\e)}} \leq \frac{\beta_f}{\lambda},
\end{equation*}
which implies
\begin{equation*}
\zeta^{(m)}_J \in 2\lambda J \left[\kappa-\frac{3\beta_f}{2\lambda},\kappa+\frac{3\beta_f}{2\lambda}\right] 
\quad \iff \quad
\kappa \in \frac{1}{2\lambda J} \left[\zeta^{(m)}_J - 3\beta_f, \zeta^{(m)}_J + 3\beta_f\right].
\end{equation*}
As $J = \frac{6\beta_f}{\lambda}$ and $\kappa \in \N$, we can conclude
\begin{equation*}
\kappa \in \left[\frac{\zeta^{(m)}_J}{12\beta_f} - \frac{1}{4}, \frac{\zeta^{(m)}_J}{12\beta_f} + \frac{1}{4}\right]
\quad \implies \quad
\kappa = \left\lfloor\frac{\zeta^{(m)}_J}{12\beta_f} + \frac{1}{2}\right\rfloor
\end{equation*}
so that indeed $\Delta^{N-m-1} \varepsilon_\theta = s_{(m+1)}$.
In particular, the choice $m = N-1$ then gives
\begin{equation*}
\Sum s_{(N-1)} = \Sum \Delta \varepsilon_\theta = \varepsilon_\theta + 2\lambda\nu
\end{equation*}
with some $\nu \in \Z$.
Recall that $\Radon_\theta f \in \PW_\Omega$ implies
\begin{equation*}
\|\Radon_\theta f\|_{\Lebesgue^2(\R)}^2 = \T \sum_{k \in \Z} |y_\theta\sqb{k}|^2 < \infty
\end{equation*}
so that $y_\theta \in c_0(\Z)$ is a null sequence.
This, in turn, implies that $\varepsilon_\theta$ forms a null sequence as well and, hence,
\begin{equation*}
\lim_{j \to \infty} \bigl(\Sum s_{(N-1)}\bigr)[j] = 2\lambda\nu.
\end{equation*}
Using the Modulo Decomposition~\eqref{eq:modulo_decomposition}, we can recover the Radon samples $y_\theta\sqb{k}$ for any $k \in \Z$ via
\begin{equation*}
y_\theta\sqb{k} = y_\theta^\lambda\sqb{k} + \varepsilon_\theta\sqb{k} = y_\theta^\lambda\sqb{k} + \bigl(\Sum s_{(N-1)}\bigr)\sqb{k} - \lim_{j \to \infty} \bigl(\Sum s_{(N-1)}\bigr)[j]
\end{equation*}
and, thus, we can reconstruct $\Radon_\theta f \in \PW_\Omega$ by applying the Shannon sampling theorem, i.e.,
\begin{equation*}
\Radon_\theta f = \sum_{k \in \Z} y_\theta\sqb{k} \sinc\left(\frac{\pi}{\T}(\cdot - k\T)\right) = p_\theta.
\end{equation*}
In this way we get access to $\Radon_\theta f$ for all $\theta \in [0,\pi)$ and, using the evenness of $\Radon f$, also for $[\pi,2\pi)$. As $f \in \Bernstein_\Omega^1(\R^2)$ is $\Omega$-band-limited and we chose the Ram-Lak filter with the same bandwidth, i.e.,
\begin{equation*}
\Fourier_1 F_\Omega(S) = |S| \, \ind_{[-\Omega,\Omega]}(S),
\end{equation*}
we can apply the FBP formula~\eqref{eq:FBP} to exactly recover $f$ via
\begin{equation*}
f = \frac{1}{2} \Radon^\# (F_\Omega * \Radon_\theta f) = \frac{1}{2} \Radon^\# (F_\Omega * p_\theta) = f_\Omega^\lambda
\end{equation*}
and the proof is complete.
\end{proof}

As discussed in Section~\ref{sec:approach}, the bandwidth assumption underlying our theory can arise in two ways.
In some cases, e.g.~in seismic imaging, it naturally comes up as part of the model, in other cases, such as X-ray tomography, a model class of compactly supported functions is assumed, and the band-limitation is enforced by the anti-aliasing filtering step.
We will now focus on the latter scenario and describe the filtering step in more detail for a {\em not necessarily} band-limited function $f \in \Lebesgue^1(\R^2)$.
\smallskip
\begin{enumerate}[label=(\roman*)]
\item
\label{itm:pre_filter}
For fixed $\theta \in [0,\pi)$ we pre-filter $\Radon_\theta f$ with the ideal low-pass filter $\Phi_\Omega \in \PW_\Omega$ satisfying
\begin{equation}
\label{eq:pre_filter}
\Fourier_1 \Phi_\Omega(S) = \ind_{[-\Omega,\Omega]}(S).
\end{equation}
The resulting band-limited Radon projection $p_\theta \in \PW_\Omega$ takes the form
\begin{equation}
\label{eq:Radon_projection}
p_\theta(t) = \int_\R \Radon_\theta f(S) \, \Phi_\Omega(t-S) \: \d S.
\end{equation}
\item The Radon Projection $p_\theta$ is folded in the range $[-\lambda,\lambda]$ via the modulo operator $\Mod_\lambda$ resulting in the Modulo Radon projection
\begin{equation}
\label{eq:Modulo_projection}
p_\theta^\lambda(t) = \Mod_\lambda(p_\theta(t)).
\end{equation}
\item The Modulo Radon Projection $p_\theta^\lambda$ is sampled with rate $\T$ yielding, for $m \in \Z$,
\begin{equation*}
y_\theta^\lambda[m] = p_\theta^\lambda(m\T) = \Mod_\lambda(p_\theta(m\T)).
\end{equation*}
\end{enumerate}
To approximately recover $f \in \Lebesgue^1(\R^2)$ from its given semi-discrete Modulo Radon Projections
\begin{equation}
\label{eq:MRT_projections_discrete}
\bigl\{p_\theta^\lambda(m\T) \bigm| \theta \in [0,\pi), ~ m \in \Z\bigr\},
\end{equation}
we apply Algorithm~\ref{algo:US_FBP} to the samples $y_\theta^\lambda\sqb{k} = p_\theta^\lambda(k\T)$ yielding the US-FBP reconstruction~$f_\Omega^\lambda$.
Note that, if $f \in \Bernstein_\Omega^1(\R^2)$ is itself band-limited with bandwidth $\Omega$, we have
\begin{equation*}
p_\theta = \Radon_\theta f
\quad \mbox{ and } \quad
p_\theta^\lambda = \Modulo_\theta^\lambda f.
\end{equation*}

\subsubsection*{A Note on US-FBP Error Bounds}

If we use a reconstruction filter $F_\Omega$ satisfying~\eqref{eq:recon_filter}, we observe that
\begin{equation*}
f_\Omega = \frac{1}{2} \Radon^\# \left(F_\Omega * \Radon_\theta f\right) = \frac{1}{2} \Radon^\# \left(F_\Omega * p_\theta\right) = f_\Omega^\lambda
\end{equation*}
and, thus, existing error estimates~\cite{Beckmann2019, Beckmann2019a, Beckmann2020a, Beckmann2021} for the FBP approximation error $f - f_\Omega$ carry over to the US-FBP error $f - f_\Omega^\lambda$.
For illustration, we apply~\cite[Theorem~5.5]{Beckmann2019} and~\cite[Theorem~3]{Beckmann2019a} to obtain error estimates in Sobolev spaces of fractional order $\alpha > 0$, given by
\begin{equation*}
\Sobolev^\alpha(\R^2) = \bigl\{f \in \Lebesgue^2(\R^2) \bigm| \|f\|_\alpha < \infty\bigr\},
\end{equation*}
where
\begin{equation*}
\|f\|_\alpha^2 = \frac{1}{4\pi^2} \int_{\R^2} (1 + \|\bfx\|_2^2)^\alpha \, |\Fourier_2 f(\bfx)|^2 \: \d \bfx.
\end{equation*}

\begin{theorem}
\label{theorem:US_FBP_error}
For $\alpha > 0$, let $f \in \Lebesgue^1(\R^2) \cap \Sobolev^\alpha(\R^2)$ and, for $\lambda > 0$, let $f_\Omega^\lambda$ denote the \mbox{US-FBP} reconstruction from semi-discrete Modulo Radon Projections~\eqref{eq:MRT_projections_discrete} with sampling rate $\T < \nicefrac{1}{\Omega\e}$.
If the window function $W$ satisfies $W \in \Cont([-1,1])$ and $W(S) = 1$ for $S \in [-c_W,c_W]$ with $c_W > 0$, then the $\Lebesgue^2$-norm of the US-FBP error is bounded above by
\begin{equation*}
\|f - f_\Omega^\lambda\|_{\Lebesgue^2(\R^2)} \leq \left(c_W^{-\alpha} \, \|1 - W\|_\infty + 1\right) \Omega^{-\alpha} \, \|f\|_\alpha.
\end{equation*}
Alternatively, if $W \in \Cont^{k-1,\nu}([-1,1])$ for $k \in \N$, $\nu \in (0,1]$ and $W(0) = 1$ as well as $W^{(j)}(0) = 0$ for all $1 \leq j \leq k-1$, then the $\Lebesgue^2$-norm of the US-FBP error is bounded by
\begin{equation*}
\|f - f_\Omega^\lambda\|_{\Lebesgue^2(\R^2)} \leq c_{\alpha,W} \, \Omega^{-\min\{k-1+\nu,\alpha\}} \, \|f\|_\alpha
\end{equation*}
with an explicitly known constant $c_{\alpha,W} > 0$.
\end{theorem}

\subsection{Approximate Recovery from Finitely Many MRT Samples}
\label{sec:reconstruction_finite}

In practice only finitely many samples can be taken and, hence, our reconstruction procedure needs to be adjusted.
We again focus on the case of a compactly supported target function $f$ whose Radon transform is subjected to the pre-filtering step~\ref{itm:pre_filter}.
More precisely, we assume that $f$ is supported in $B_1(0) \subset \R^2$, i.e.,
\begin{equation*}
f(\bfx) = 0
\quad \forall \, \|\bfx\|_2 > 1.
\end{equation*}
In this case, the Radon projection $p_\theta \in \PW_\Omega$ in~\eqref{eq:Radon_projection} produced in the pre-filtering step~\ref{itm:pre_filter} differs from $\Radon_\theta f$ and will no longer have  compact support, but will decay at infinity in the sense that for any $c > 0$ there exists $t_c > 0$ such that
\begin{equation*}
|p_\theta(t)| < c
\quad \forall \, |t| > t_c.
\end{equation*}

Moreover, we assume that we are given {\em fully discrete} Modulo Radon Projections
\begin{equation}
\label{eq:MRT_samples_fully_discrete}
\bigl\{p_{\theta_m}^\lambda(t_k) \bigm| -K' \leq k \leq K, ~ 0 \leq m \leq M-1\bigr\}
\end{equation}
in parallel beam geometry with
\begin{equation*}
t_k = k \, \T
\quad \mbox{ and } \quad
\theta_m = m \, \frac{\pi}{M},
\end{equation*}
where $\T > 0$ is the spacing of $K'+K+1$ parallel lines per angle.
To adapt our reconstruction algorithm to this fully discrete scenario, we again follow a sequential reconstruction approach.
First, we devise an Unlimited Sampling (US) algorithm for essentially compactly supported band-limited functions (more precisely, for the closely related class of functions with compact $\lambda$-exceedance, cf.~Definition~\ref{def:lambda_support}), see Algorithm~\ref{algo:UL_lambda}, and apply it to recover the Radon projection $p_{\theta_m}$ from $p_{\theta_m}^\lambda$ for each angle $\theta_m$.
In a second step, we then reconstruct $f$ from the recovered Radon projections
\begin{equation*}
\bigl\{p_{\theta_m}(t_k) \bigm| -K' \leq k \leq K, ~ 0 \leq m \leq M-1\bigr\}
\end{equation*}
by applying a discrete version of the approximate FBP formula~\eqref{eq:FBP}, see Algorithm~\ref{algo:FBP}.

\bigskip

Both stages of our proposed method will be described in detail in the upcoming paragraphs.

\subsubsection*{Unlimited Sampling for Functions of Compact $\lambda$-Exceedance}

The goal of this section is to provide an algorithm that is guaranteed to recover finitely many samples $\gamma\sqb{k} = g(k\T)$ of a function $g$ with sampling rate $\T > 0$, given the modulo samples $y\sqb{k} = \Mod_\lambda(g(k\T))$.
Motivated by the considerations of the previous paragraph, we will consider $g \in \PW_\Omega$, i.e.,  $\Omega$-band-limited functions, which  in addition exceed the modulo threshold only in a compact region (this is a special case of the decay property motivated above), cf.~the following definition.

\begin{definition}[$\lambda$-exceedance property]
\label{def:lambda_support}
Let $\lambda > 0$.
We say that a univariate function $g: \R \longrightarrow \R$ is of {\em compact $\lambda$-exceedance with parameter $\rho > 0$} if $|g(t)| < \lambda$ for $|t| > \rho$.
We then write~$g \in \Band^\lambda_\rho$.
\end{definition}

Our recovery strategy involves the forward difference operator $\Delta: \R^{K'+K+1} \longrightarrow \R^{K'+K}$ with
\begin{align*}
(\Delta a)\sqb{k} = a[k+1]-a\sqb{k}
\quad \mbox{ for } -K' \leq k < K
\end{align*}
and the anti-difference operator $\Sum: \R^{K'+K} \longrightarrow \R^{K'+K+1}$ with
\begin{equation*}
(\Sum a')\sqb{k} = \sum_{j=-K'}^{k-1} a'[j]
\quad \mbox{ for } -K' \leq k \leq K
\end{equation*}
so that
\begin{equation*}
\Sum(\Delta a) = a - a\sqb{-K'}
\quad \forall \, a = (a\sqb{k})_{k=-K'}^K \in \R^{K'+K+1}.
\end{equation*}
Furthermore, we again use the modulo decomposition from Proposition~\ref{proposition:modulo_decomposition},
\begin{equation*}
g(t) = \Mod_\lambda(g(t)) + \varepsilon_g(t),
\end{equation*}
where $\varepsilon_g$ is a piecewise constant function with values in $2\lambda\Z$.
For the sake of brevity, we set
\begin{equation*}
\varepsilon_\gamma\sqb{k} = \varepsilon_g(k\T) = \gamma\sqb{k} - y\sqb{k}
\quad \mbox{ for } -K' \leq k \leq K.
\end{equation*}

\begin{algorithm}[t]
\caption{Unlimited Sampling for functions of compact $\lambda$-exceedance}
\label{algo:UL_lambda}
\begin{algorithmic}[1]
\REQUIRE modulo samples $y\sqb{k} = \Mod_\lambda(g(k\T))$ for $k=-K',\ldots,K$ and upper bound $\beta_g \geq \|g\|_\infty$
\medskip
\STATE \textbf{choose} $N = \left\lceil\frac{\log(\lambda) - \log(\beta_g)}{\log(\T\Omega\e)}\right\rceil_+$
\smallskip
\STATE $s_{(0)}\sqb{k} = (\Delta^N \varepsilon_\gamma)\sqb{k} = (\Mod_\lambda(\Delta^N y) - \Delta^N y)\sqb{k}$
\FOR{$n=0,\ldots,N-2$}
\STATE $s_{(n+1)}\sqb{k} = 2\lambda \left\lceil\frac{\lfloor\nicefrac{\Sum s_{(n)}\sqb{k}}{\lambda}\rfloor}{2}\right\rceil$
\ENDFOR
\STATE $\gamma\sqb{k} = y\sqb{k} + (\Sum s_{(N-1)})\sqb{k}$
\medskip
\ENSURE samples $\gamma\sqb{k} = g(k\T)$ for $k=-K,\ldots,K$
\end{algorithmic}
\end{algorithm}

\begin{theorem}[Sampling theorem for functions of compact $\lambda$-exceedance]
\label{theo:UL_lambda}
Let $g \in \PW_\Omega \cap \Band_\rho^\lambda$ be a band-limited function of compact $\lambda$-exceedance with parameter $\rho$.
Let $\beta_g > 0$ be given with $\|g\|_\infty \leq \beta_g$.
Then, a sufficient condition for the {\em exact} recovery of the samples $\gamma\sqb{k} = g(k\T)$, $-K \leq k \leq K$, from given modulo samples $y\sqb{k} = \Mod_\lambda(g(k\T))$, $-K' \leq k \leq K$, by means of Algorithm~\ref{algo:UL_lambda} is given by
\begin{equation*}
T < \frac{1}{\Omega\e}
\quad \mbox{ and } \quad
K' \geq \max\{K, \rho \T^{-1}+N\},
\end{equation*}
where
\begin{equation*}
N = \left\lceil\frac{\log(\lambda) - \log(\beta_g)}{\log(\T\Omega\e)}\right\rceil_+.
\end{equation*}
\end{theorem}

\begin{proof}
If $\beta_g \leq \lambda$, we have $y\sqb{k}=\Mod_\lambda(g(k\T)) = g(k\T) = \gamma\sqb{k}$ and the statement is trivially true.
Thus, we only have to consider the case $\beta_g > \lambda$.
By the choice of $\T$ and $N$ we have
\begin{equation*}
N = \left\lceil\frac{\log(\lambda) - \log(\beta_g)}{\log(\T\Omega\e)}\right\rceil_+ \geq \frac{\log\bigl(\frac{\lambda}{\beta_g}\bigr)}{\log(\T\Omega\e)}
\quad \mbox{ and } \quad
\log(\T\Omega\e) \leq -1.
\end{equation*}
This implies
\begin{equation*}
(\T\Omega\e)^N \leq \frac{\lambda}{\beta_g}
\end{equation*}
and, consequently, Lemma~\ref{lemma:difference_norm} ensures that
\begin{equation*}
\|\Delta^N \gamma\|_\infty \leq \frac{\lambda}{\beta_g} \, \|g\|_\infty \leq \lambda,
\end{equation*}
which in turn gives
\begin{equation*}
\Delta^N \gamma = \Mod_\lambda(\Delta^N \gamma) = \Mod_\lambda(\Delta^N y).
\end{equation*}
Thus, for $\varepsilon_\gamma\sqb{k} = \varepsilon_g(k\T) = \gamma\sqb{k} - y\sqb{k}$ follows that
\begin{equation*}
\Delta^N \varepsilon_\gamma = \Delta^N(\gamma-y) = \Delta^N \gamma - \Delta^N y = \Mod_\lambda(\Delta^N y) - \Delta^N y
\end{equation*}
and $(\Delta^N \varepsilon_\lambda)\sqb{k}$, $-K' \leq k \leq K$, can be computed from the modulo samples $y\sqb{k}$, $-K' \leq k \leq K$.
Since $g$ has compact $\lambda$-exceedance with parameter $\rho$ and $K'$ satisfies $K' \geq \rho \T^{-1} + N$, we have
\begin{equation*}
\varepsilon_\gamma[-K'] = g(-K'\T) - \underbrace{\Mod_\lambda(g(-K'\T))}_{= g(-K'\T)} = 0
\end{equation*}
and, for all $1 \leq n \leq N$,
\begin{align*}
(\Delta^n \varepsilon_\gamma)[-K'] & = \sum_{m=0}^n \binom{n}{m} \, (-1)^{n-m} \, \varepsilon_\gamma[-K'+m] \\
& = \sum_{m=0}^n \binom{n}{m} \, (-1)^{n-m} \, (g((-K'+m)\T) - \underbrace{\Mod_\lambda(g((-K'+m)\T)))}_{= g((-K'+m)\T)} = 0.
\end{align*}
With this, we show by induction in $j \in \{0,\ldots,N-1\}$ that $s_{(j)} = \Delta^{N-j} \varepsilon_\gamma \in \R^{K'+K+1-N+j}$.
The induction seed $j = 0$ reduces to the definition of $s_{(0)} = \Delta^N \varepsilon_\gamma$.
For the induction step, assume that the induction hypothesis holds for $j \in \{0,\ldots,N-2\}$, i.e.,
\begin{equation*}
s_{(j)} = \Delta^{N-j} \varepsilon_\gamma = \Delta(\Delta^{N-(j+1)} \varepsilon_\gamma).
\end{equation*}
Then, applying the anti-difference operator $\Sum$ gives
\begin{equation*}
\Sum s_{(j)} = \Sum \bigl(\Delta(\Delta^{N-(j+1)} \varepsilon_\gamma)\bigr) = \Delta^{N-(j+1)} \varepsilon_\gamma - \underbrace{\Delta^{N-(j+1)} \varepsilon_\gamma[-K']}_{=0} = \Delta^{N-(j+1)} \varepsilon_\gamma.
\end{equation*}
In particular, we have $(\Sum s_{(j)})\sqb{k} \in 2\lambda\Z$ and indeed obtain
\begin{equation*}
s_{(j+1)} = 2\lambda \left\lceil\frac{\lfloor\nicefrac{\Sum s_{(j)}}{\lambda}\rfloor}{2}\right\rceil = \Sum s_{(j)} = \Delta^{N-(j+1)} \varepsilon_\gamma.
\end{equation*}
Choosing $j = N-1$ yields $s_{(N-1)} = \Delta \varepsilon_\gamma$ and, consequently,
\begin{equation*}
\Sum s_{(N-1)} = \Sum(\Delta \varepsilon_\gamma) = \varepsilon_\gamma - \underbrace{\varepsilon_\gamma[-K']}_{=0} = \varepsilon_\gamma.
\end{equation*}
This in combination with the modular decomposition property ensures that
\begin{equation*}
\gamma\sqb{k} = y\sqb{k} + (\Sum s_{(N-1)})\sqb{k}
\end{equation*}
and Algorithm~\ref{algo:UL_lambda} exactly recovers the samples $\gamma\sqb{k}$ of $g$ from the modulo samples $y\sqb{k}$.
\end{proof}

\subsubsection*{Discrete FBP Reconstruction Formula for Parallel Beam Geometry}

We now address the discretization of the FBP formula~\eqref{eq:FBP} for the approximate reconstruction of a compactly supported function $f \in \Lebesgue^1(\R^2)$ from discrete Radon data in parallel beam geometry
\begin{equation}
\label{eq:RT_samples_finite}
\bigl\{\Radon_{\theta_m} f(t_k) \bigm| -K \leq k \leq K, ~ 0 \leq m \leq M-1\bigr\}
\end{equation}
with $t_k = k \, \T$ and $\theta_m = m \, \frac{\pi}{M}$, where $\T > 0$ is the spacing of $2K+1$ parallel lines per angle.

We follow a standard approach~\cite{Natterer2001} and apply the composite trapezoidal rule to discretize the convolution $\ast$ and back projection $\Radon^\#$.
This leads to the discrete reconstruction formula
\begin{equation*}
f_D(\bfx) = \frac{\T}{2M} \sum_{m=0}^{M-1} \sum_{k = -K}^K F_\Omega(x_1\cos(\theta_m)+x_2\sin(\theta_m) - t_k) \, \Radon_{\theta_m} f(t_k)
\quad \mbox{ for } \bfx \in \R^2,
\end{equation*}
in short,
\begin{equation*}
f_D = \frac{1}{2} \Radon^\#_D \bigl(F_\Omega *_D \Radon f\bigr).
\end{equation*}
Note that the evaluation of the discrete reconstruction $f_D$ requires the computation of the values
\begin{equation*}
(F_\Omega *_D \Radon_{\theta_m} f)(x_1\cos(\theta_m)+x_2\sin(\theta_m))
\quad \forall \, 0 \leq m \leq M-1
\end{equation*}
for each reconstruction point $\bfx = (x_1,x_2) \in \R^2$.
To reduce the computational costs, one typically evaluates the function
\begin{equation*}
h_{\theta_m}(t) = (F_\Omega *_D \Radon_{\theta_m} f)(t)
\quad \mbox{ for } t \in \R
\end{equation*}
only at the points $t = t_i$, $i \in I$, and interpolates the value $h_{\theta_k}(t)$ for $t = x_1\cos(\theta_m)+x_2\sin(\theta_m)$ using linear spline interpolation $\Int_1$.
This leads us to the {\em discrete FBP reconstruction formula}
\begin{equation}
\label{eq:discrete_FBP_formula}
f_\FBP = \frac{1}{2} \Radon^\#_D \bigl(\Int_1[F_\Omega *_D \Radon f]\bigr),
\end{equation}
which is summarized in Algorithm~\ref{algo:FBP}.
According to~\cite[Section 5.1.1]{Natterer2001}, the optimal sampling conditions are given by
\begin{equation*}
\T \leq \frac{\pi}{\Omega}, \quad K \geq \frac{1}{\T}, \quad M \geq \Omega.
\end{equation*}

\begin{algorithm}[t]
\caption{Discrete filtered back projection in parallel beam geometry}
\label{algo:FBP}
\begin{algorithmic}[1]
\REQUIRE Radon samples $\Radon_{\theta_m} f(t_k)$ with radial spacing~$\T$ for $k=-K,\ldots,K$, $m=0,\ldots,M-1$, reconstruction filter $F_\Omega \in \PW_\Omega$
\medskip
\FOR{$m=0,\ldots,M-1$}
\FOR{$i \in I$}
\STATE ${\displaystyle h_{\theta_m}(t_i) = \sum_{k=-K}^K F_\Omega(t_i - t_k) \, \Radon_{\theta_m} f(t_k)}$
\ENDFOR
\ENDFOR
\medskip
\STATE ${\displaystyle f_\FBP(\bfx) = \frac{\T}{2M} \sum_{m=0}^{M-1} \Int_1 h_{\theta_m}(x_1 \cos(\theta_m) + x_2 \sin(\theta_m))}$
\medskip
\ENSURE discrete FBP reconstruction $f_\FBP$
\end{algorithmic}
\end{algorithm}

\subsubsection*{Discrete US-FBP Reconstruction}

Combining Algorithm~\ref{algo:UL_lambda} and Algorithm~\ref{algo:FBP} to a sequential reconstruction scheme allows us to approximately recover a function $f \in \Lebesgue^1(\R^2)$ with $\supp(f) \subseteq B_1(0)$ from its finitely many Modulo Radon Projections
\begin{equation}
\label{eq:MRT_projections_finite}
\bigl\{p_{\theta_m}^\lambda(t_k) \bigm| -K' \leq k \leq K, ~ 0 \leq m \leq M-1\bigr\}
\end{equation}
with
\begin{equation*}
t_k = k \, \T
\quad \mbox{ and } \quad
\theta_m = m \, \frac{\pi}{M}
\end{equation*}
if we choose the sampling parameters
\begin{equation*}
\T < \frac{1}{\Omega\e}, \quad
K' \geq \max\{K, \rho \T^{-1}+N\}, \quad
K \geq \frac{1}{\T}, \quad
M \geq \Omega
\end{equation*}
with
\begin{equation*}
N = \left\lceil\frac{\log(\lambda) - \log(\beta)}{\log(\T\Omega\e)}\right\rceil_+,
\end{equation*}
where the uniform constants $\rho, \beta > 0$ are chosen such that
\begin{equation*}
p_{\theta_m} \in \Band^\lambda_\rho
\quad \mbox{ and } \quad
\|p_{\theta_m}\|_\infty \leq \beta
\end{equation*}
for all $0 \leq m \leq M-1$.
Indeed, applying Algorithm~\ref{algo:UL_lambda} with input data~\eqref{eq:MRT_projections_finite} yields the Radon Projections
\begin{equation}
\label{eq:RT_projections_finite}
\bigl\{p_{\theta_m}(t_k) \bigm| -K \leq k \leq K, ~ 0 \leq m \leq M-1\bigr\}
\end{equation}
and applying Algorithm~\ref{algo:FBP} with input data~\eqref{eq:RT_projections_finite} yields an approximate reconstruction of~$f$, which we call {\em discrete US-FBP reconstruction} and denote by $f^\lambda_\FBP$.

Let us again stress that using the Radon Projections $p_\theta$ in FBP formula~\eqref{eq:FBP} yields the same results as using the true Radon data $\Radon_\theta f$.
This is, because we use a reconstruction filter $F_\Omega \in \PW_\Omega$ satisfying~\eqref{eq:recon_filter} in FBP formula~\eqref{eq:FBP} and the ideal low-pass filter $\Phi_\Omega \in \PW_\Omega$ satisfying~\eqref{eq:pre_filter} in the pre-filtering step~\ref{itm:pre_filter}.

We finally remark that instead of applying Algorithm~\ref{algo:UL_lambda} one could think of adjusting Algorithm~\ref{algo:US_FBP} to the case of finitely many MRT projections.
In this case we would need at least
\begin{equation}
\label{eq:alg31_samples}
\max\{2K+1,J+N\}
\end{equation}
many samples per angle $\theta_m$, where $J = 6\frac{\beta_f}{\lambda}$ is independent of the bandwidth $\Omega$.
In contrast to this, Algorithm~\ref{algo:UL_lambda} needs
\begin{equation}
\label{eq:alg32_samples}
\max\{2K+1,\rho\T^{-1}+K+1+N\}
\end{equation}
many samples per angle $\theta_m$, where the parameter $\rho$ depends on the threshold $\lambda$ and on the bandwidth $\Omega$ due to the pre-filtering step~\ref{itm:pre_filter}.
For a sufficiently large $\Omega > 0$ and small $\lambda > 0$ the increase in sample size due to Algorithm~\ref{algo:UL_lambda} can be significantly smaller than the increase incurred by adjusting Algorithm~\ref{algo:US_FBP}.
To illustrate this, we choose $K$ of minimal order $K \sim \Omega$ and consider the following benchmark scenario of Sobolev functions of fractional order $\alpha > 1$, which has also been studied in~\cite{Rieder2007} in the context of optimal convergence rates for filtered back projection reconstructions from finitely many Radon samples.
Let $f \in \Sobolev^\alpha(\R^2)$ with $\alpha > 1$ and $\supp(f) \subseteq B_1(0)$.
Then, for all $\theta \in [0,\pi)$ we have $\supp(\Radon_\theta f) \subseteq [-1,1]$ and $\Radon_\theta f \in \Sobolev^{\alpha+\nicefrac{1}{2}}(\R)$ so that
\begin{equation*}
|\Fourier(\Radon_\theta f)(\omega)| \sim |\omega|^{-(\alpha+1)}
\quad \mbox{ for } |\omega| \to \infty.
\end{equation*}
Consequently, for the Radon projection $p_\theta \in \PW_\Omega$ follows that
\begin{align*}
\|p_\theta\|_{\Lebesgue^\infty(\R\setminus[-1,1])} & = \|\Radon_\theta f - p_\theta\|_{\Lebesgue^\infty(\R\setminus[-1,1])} \lesssim \|\Fourier(\Radon_\theta f - p_\theta)\|_{\Lebesgue^1(\R)} \\
& \lesssim  \|\Fourier (\Radon_\theta f)\|_{\Lebesgue^1(\R\setminus[-\Omega,\Omega])} \lesssim \int_\Omega^\infty |\omega|^{-(\alpha+1)} \: \d \omega \lesssim \Omega^{-\alpha}
\end{align*}
and, if $\lambda \sim \Omega^{-\alpha}$, we obtain $\rho \approx 1$.
In this scenario, we have
\begin{equation*}
J \sim \lambda^{-1} \sim \Omega^\alpha
\quad \mbox{ and } \quad
\rho\T^{-1}+K+1 \sim \Omega
\end{equation*}
so that by \eqref{eq:alg31_samples} and \eqref{eq:alg32_samples}, the number of samples required per angle by Algorithms~\ref{algo:US_FBP} and \ref{algo:UL_lambda} is of order $\Omega^\alpha$ and $\Omega$, respectively.
Consequently, as $\alpha > 1$, applying Algorithm~\ref{algo:UL_lambda} is more efficient than employing an adjusted version of Algorithm~\ref{algo:US_FBP}.
Note that the above analysis indicated that the advantage of Algorithm~\ref{algo:UL_lambda} is more pronounced for large smoothness parameters $\alpha > 1$.
While this scenario is only a first benchmark and excludes applications in medical imaging, in which case we typically assume $\alpha < \frac{1}{2}$ to allow for discontinuities along smooth curves,
 our numerical experiments with the Shepp-Logan phantom, as reported in the next section, demonstrate the superiority of Algorithm~\ref{algo:UL_lambda} also in the low smoothness regime.
Hence, the above analysis has to be understood as a first step in explaining the advantages of compact $\lambda$-exceedance and calls for a deeper investigation in future work.


\section{Numerical Assessment}
\label{sec:numerics}

The purpose of our numerical experiments is to demonstrate the single shot, HDR reconstruction approach based on the MRT. Starting with the benchmark ``Shepp-Logan Phantom'' that serves as a simulation study, we use the \emph{open source} ``Walnut Dataset'' \cite{Siltanen2015} that includes realistic uncertainties arising from the tomography hardware. This dataset is then used in two different configurations. In the first case, we apply modulo non-linearity in MATLAB and demonstrate our HDR recovery approach. In the second case, we generate the one-dimensional Radon projections included in the Walnut data set as analog signals
and physically acquire modulo samples using our custom designed, prototype modulo ADC. This serves as an experimental example for testing the capability of our approach in a realistic setting and shows that our method can also deal with hardware imperfections. At the same it, the sequential setup clearly indicates that 
existing tomography equipment can be readily augmented with our prototype ADC. 
Our final demonstration highlights an interesting aspect of the recovery algorithms, namely, reconstruction is possible with slower sampling rates than what is specified by Theorem~\ref{theorem:US_FBP}, indicating the potential for developing tighter bounds. 

\subsection{Experimental Demonstration for Synthetic Data}

\subsubsection*{Shepp-Logan Phantom}

In a first set of numerical experiments we use the proposed US-FBP framework to recover the Shepp-Logan phantom~\cite{Shepp1974} on a grid of $256 \times 256$ pixels from finitely many Modulo Radon Projections~\eqref{eq:MRT_projections_finite}.
The results are summarized in Figure~\ref{fig:US-FBP_Shepp-Logan}, where we use the parameters
\begin{equation*}
\Omega = 300, \quad \T = \frac{1}{2\Omega\e}, \quad K = \Bigl\lceil\frac{1}{\T}\Bigr\rceil, \quad M = \Omega
\end{equation*}
and as reconstruction filter the cosine filter given by
\begin{equation*}
\Fourier_1 F_\Omega(\omega) = |\omega| \, \cos\Bigl(\frac{\pi \omega}{2\Omega}\Bigr) \, \ind_{[-\Omega,\Omega]}(\omega)
\quad \mbox{ for } \omega \in \R.
\end{equation*}
As predicted by our theory, the FBP reconstruction from conventional Radon data and the US-FBP reconstruction from Modulo Radon Projections yield essentially the same root mean square error (RMSE) and the reconstructions are also basically visually indistinguishable.
This is the case  for both $\lambda = 0.025$, corresponding to a $10$-times compression in dynamic range, and for $\lambda = 0.00025$, yielding a compression factor of $1000$, so the reconstruction error does not seem to be affected by the value of~$\lambda$. This behaviour is expected as our error bounds in Theorem~\ref{theorem:US_FBP_error} are independent of $\lambda$.
 We nevertheless find it remarkable given that the sinusoidal structures that  naturally arise in the Radon domain  -- and which are clearly visible Figures~\ref{fig:US-FBP_Shepp-Logan} (a) and (b) -- are no longer visible in Figure~\ref{fig:US-FBP_Shepp-Logan} (c).
For $\lambda = 0.025$ we find $K' = 1631 = K$, so that no additional samples are needed for the application of Algorithm~\ref{algo:UL_lambda}. For $\lambda = 0.00025$ we obtain $K' = 3793$, which corresponds to $2162$ additional samples per angle. In contrast to this, for Algorithm~\ref{algo:US_FBP} we would need $J = 13320$ and $N = 12$, which corresponds to $10071$ additional samples per angle -- nearly $5$-times the amount.

\begin{figure}[!t]
\centering
\includegraphics[width = 1\textwidth]{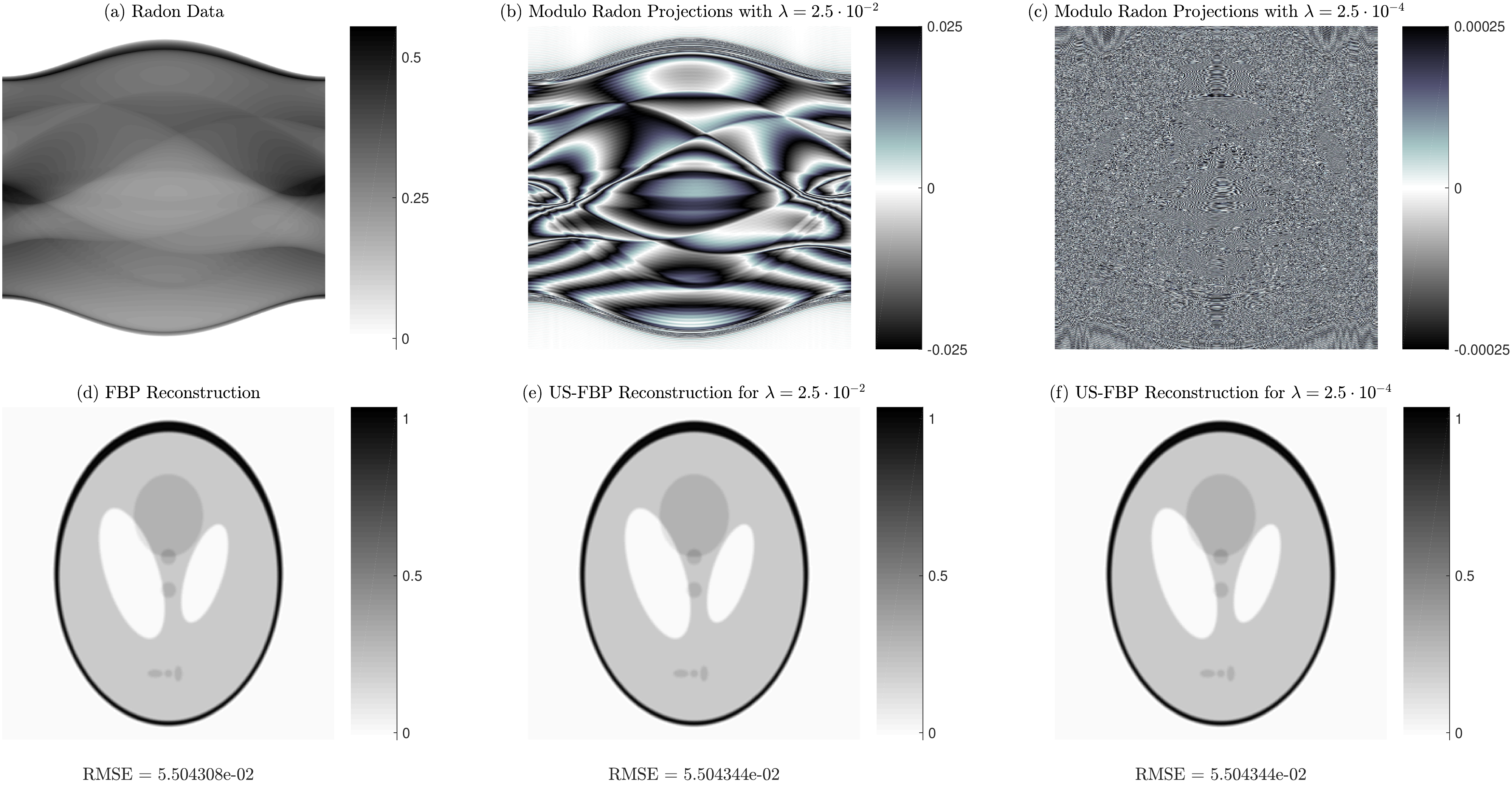}
\caption{Demonstration of US-FBP reconstruction for the Shepp-Logan phantom. (a)~True Radon data. (b)~MRT projections with $\lambda = 0.025$. (c)~MRT projections with $\lambda = 0.00025$. (d)~FBP on Radon data in (a). (e)~US-FBP on MRT data in (b). (f)~US-FBP on MRT data in (c).}
\label{fig:US-FBP_Shepp-Logan}
\end{figure}

\subsubsection*{Walnut Data} 

In a second set of experiments, we consider real Radon data that is affected by numerically applying an anti-aliasing filter followed by a modulo operation.
Our findings confirm that our approach is robust to  realistic uncertainties added in the sampling pipeline due to a practical implementation.
To this end, we consider the Walnut dataset from~\cite{Siltanen2015}, which is transformed to parallel beam geometry with $M = 600$ and $K = 1128$ corresponding to $\T = \nicefrac{1}{1128}$.
Moreover, the Radon data is normalized to the dynamical range $[0,1]$ so that $\|\Radon f\|_\infty = 1$, see Figure~\ref{fig:US-FBP_walnut}(a).
Its (simulated) Modulo Projections are displayed in Figure~\ref{fig:US-FBP_walnut}(b) for $\lambda = 0.025$ and in Figure~\ref{fig:US-FBP_walnut}(c) for $\lambda = 0.00025$.
In both cases we use $\Omega = 300$ so that $\T < \nicefrac{1}{\Omega\e}$ is fulfilled.

The reconstructions with our proposed US-FBP method are shown in Figures~\ref{fig:US-FBP_walnut}(e) and~\ref{fig:US-FBP_walnut}(f), where we again use the cosine reconstruction filter.
In both cases we observe that our algorithm yields a reconstruction of the walnut that is again visually indistinguishable  from the FBP reconstruction from conventional Radon data, cf.~Figure~\ref{fig:US-FBP_walnut}(d), while compressing the dynamic range by about $10$ and $1000$ times, respectively.

\begin{figure}[!t]
\centering
\includegraphics[width = 1\textwidth]{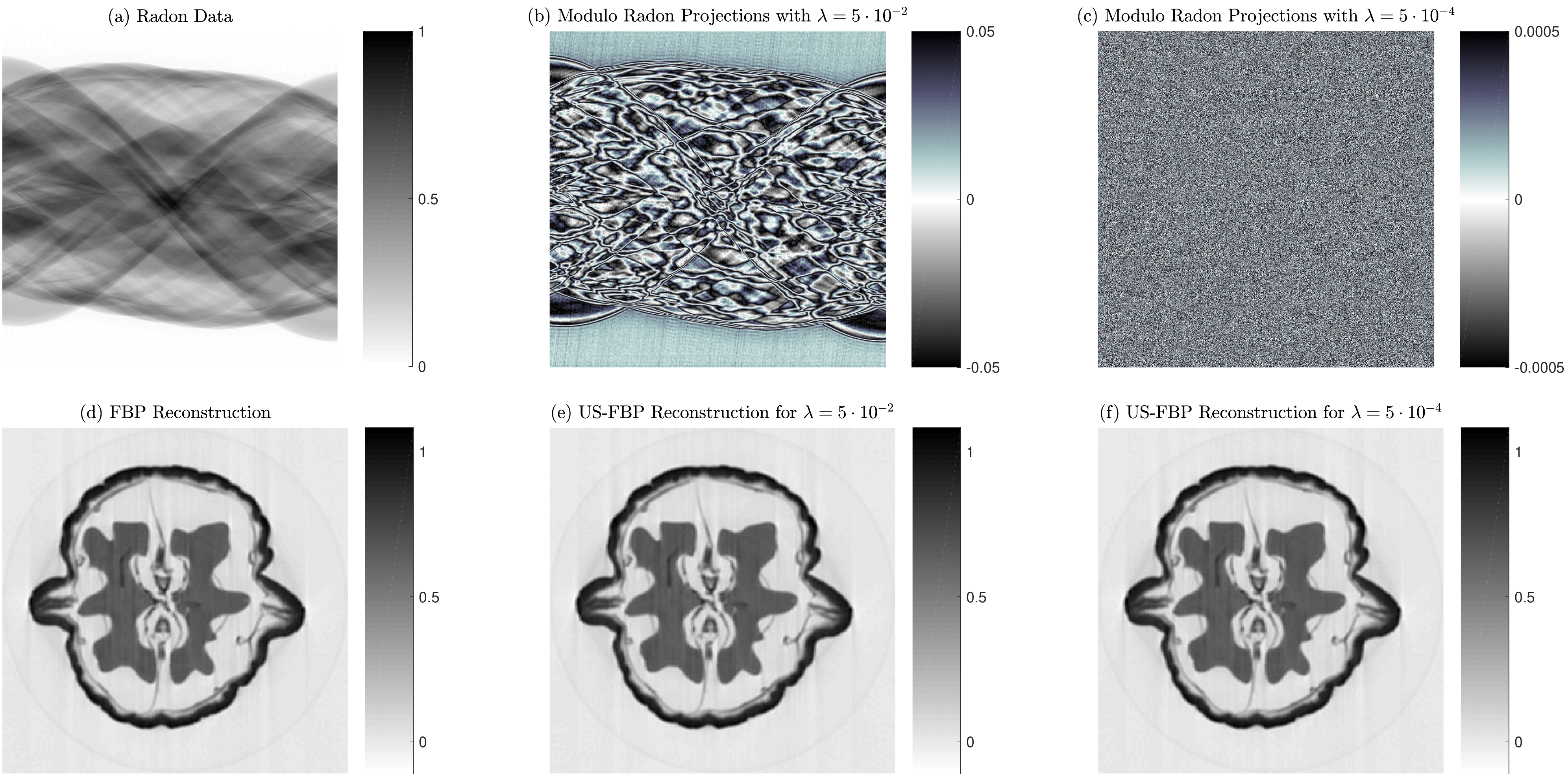}
\caption{Demonstration of US-FBP reconstruction for the Walnut dataset. (a)~Normalized Radon data. (b)~MRT projections with $\lambda = 0.05$. (c)~MRT projections with $\lambda = 0.0005$. (d)~FBP on Radon data in (a). (e)~US-FBP on MRT data in (b). (f)~US-FBP on MRT data in (c).}
\label{fig:US-FBP_walnut}
\end{figure}

\subsection{Hardware Experiment with Prototype ADC}

To further evaluate the feasibility of our work in a realistic setting, we consider the Walnut dataset from~\cite{Siltanen2015} from the previous subsection and physically digitize it using our prototype modulo ADC that converts a continuous function into modulo samples. To set up our experiment, we start with the Walnut Radon Transform projections $p_\theta = \Radon_\theta f$ at a given angle $\theta$ and physically interpolate these measurements using an electronic function generator to obtain a continuous-time analog signal $\rho_\theta$. We then shift and rescale the signal $\rho_\theta$ to yield a minimal amplitude of $0$ and a maximal amplitude of $20$ physical units (Volts), 
as the maximal amplitude of the walnut data set is close to the modulo threshold of our implementation, so without rescaling hardly any folds would happen.
We then digitize $\rho_\theta\in [0, 20]$ using two parallel sampling pipelines. 
\begin{enumerate}
  \item We sample $\rho_\theta(t)$ using a conventional ADC giving us $\rho_\theta(m\T)$. Figures~\ref{fig:usadc}-(i,a),(ii,a) and (iii,a), corresponding to $3$ different choices of projection angles $\theta$, illustrate that these samples are a good approximation of the original Walnut Radon Transform projections, up to a mean squared error of $10^{-2}$. As this small error, however, arises in the digital-to-analog conversion, hence outside of the method studied in this paper, we will use the samples of $\rho_\theta$ as our ground truth so that we can distinguish that error from the US-FBP reconstruction error.

  \item With $\lambda = 2.01$, we also sample the signal $\rho_\theta(t)$ using our prototype modulo ADC such that the modulo samples $\rho^\lambda_\theta(m\T)$ are in the range $[0,4.5]$. The MRT samples are shown in Figure~\ref{fig:usadc}-(i,b),(ii,b) and (iii,b), respectively. 
\end{enumerate}

For each of the $3$ cases shown in Figure~\ref{fig:usadc}, starting with $K = 665$ modulo samples acquired with sampling rate $\T = 75~\mu\mathrm{s}$, the signal's effective bandwidth is observed to be approximately $1000$ Hz. Note that this band-limitation is not due to the ADC that we apply when sampling $\rho_\theta$, but already present in the original data set \cite{Siltanen2015}. Using our recovery method, we obtain an HDR reconstruction with mean squared error of at most $10^{-3}$ (in each case). 

To further test our approach, we downsample the signal in Figure~\ref{fig:usadc}-(iii) by a factor of $2$ giving us $K = 333$ samples and $T = 150~\mu\mathrm{s}$. Note that due to the downsampling, reconstruction with $N=1$ fails because \eqref{eq:ybound} does not hold. To cope up with this issue, we reapply our recovery algorithm with a finite-difference order of $N=2$ and obtain near-perfect reconstruction. The mean squared error is observed to be $1.12\times10^{-3}$. The numerical metrics based on the hardware experiments, for both the non-subsampled case as well as the downsampled case, demonstrate the applicability of our approach in a real-world context. 

\begin{figure}[!t]
\centering
\includegraphics[width = 0.975\textwidth]{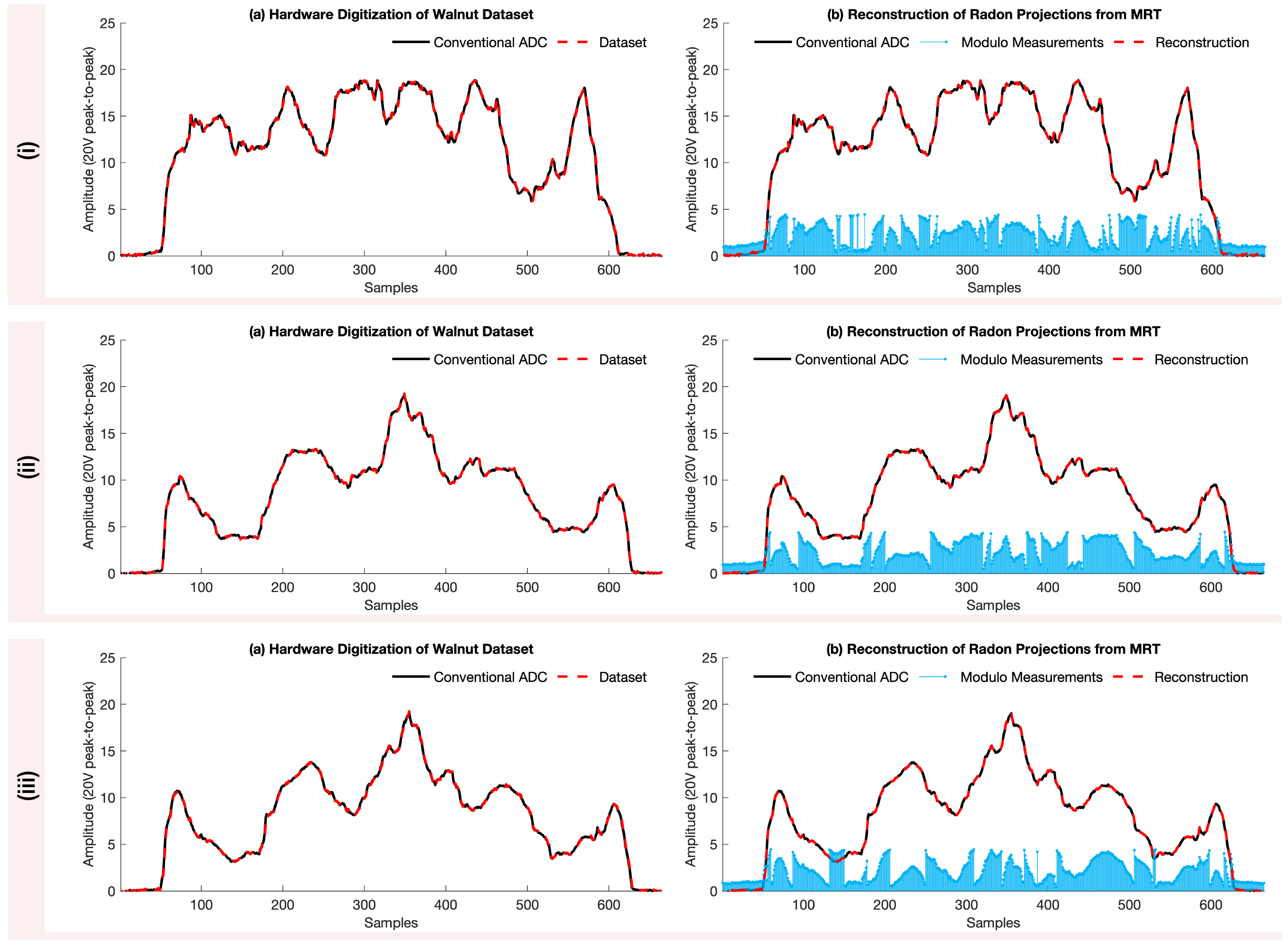}
\caption{Proof-of-concept reconstruction with our prototype modulo sampling hardware. (a) Walnut data set together with conventional, digital measurements or samples. (b) Modulo samples and reconstruction using our algorithm. In each case, our reconstruction matches the conventional digital measurements in (a) upto a mean squared error of $10^{-3}$.}
\label{fig:usadc}
\end{figure}

\subsection{Exploring the Tightness of Sampling Guarantees}

According to~\cite[Section 5.1.1]{Natterer2001}, the optimal radial sampling condition for conventional Radon data is given by
\begin{equation*}
\T \leq \frac{\pi}{\Omega} = \T_{\mathrm{Shannon}},
\end{equation*}
where the latter corresponds to the Nyquist rate for the Radon projection $p_\theta \in \PW_\Omega$ in~\eqref{eq:Radon_projection}.
However, to theoretically guarantee that Algorithm~\ref{algo:UL_lambda} recovers the Radon Projections
\begin{equation*}
\bigl\{p_{\theta_m}(k\T) \bigm| -K \leq k \leq K, ~ 0 \leq m \leq M-1\bigr\}
\end{equation*}
from given Modulo Radon Projections
\begin{equation*}
\bigl\{p_{\theta_m}^\lambda(k\T) \bigm| -K' \leq k \leq K, ~ 0 \leq m \leq M-1\bigr\},
\end{equation*}
Theorem~\ref{theo:UL_lambda} requires that the sampling rate must be a factor of $\pi\e$ faster than the Nyquist rate.
We now set up a demonstration that shows that Algorithm~\ref{algo:UL_lambda} succeeds even when the sampling rate is much slower than what is prescribed by Theorem~\ref{theo:UL_lambda}. 

\begin{figure}[t]
\centering
\includegraphics[width = 1\textwidth]{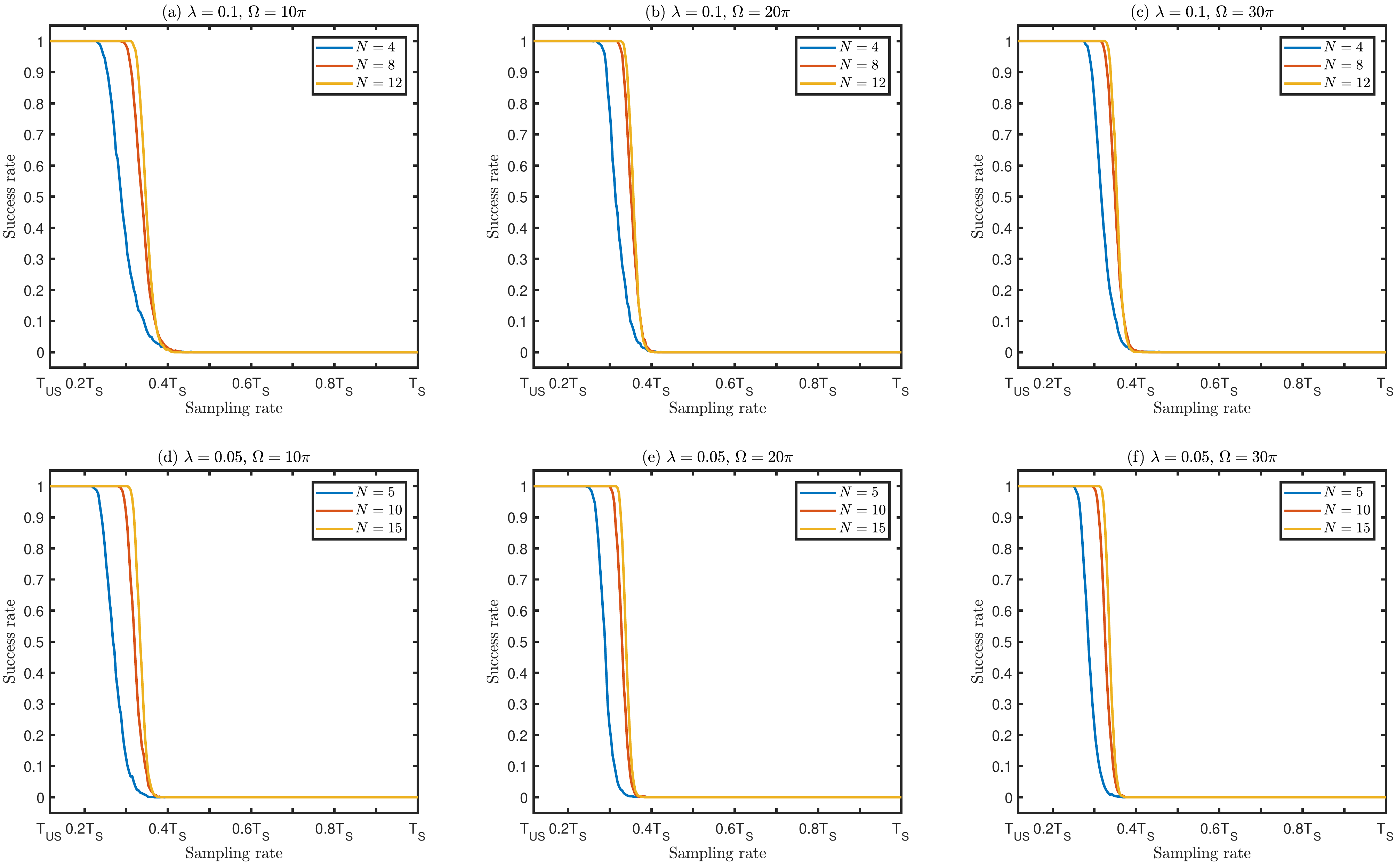}
\caption{Assessment of the tightness of the sampling guaranty $\T < \T_{\mathrm{US}}$.}
\label{fig:US_lambda_success}
\end{figure}

For $\lambda \in \{0.1,0.05\}$ and $\Omega \in \{10\pi,20\pi,30\pi\}$, we generate $1000$ realizations of a function $g \in \Band_\Omega$ with compact $\lambda$-exceedance by convolving a piecewise constant function taking values in $[-1,1]$ chosen uniformly at random with the ideal low-pass filter $\Phi_\Omega \in \PW_\Omega$ satisfying~\eqref{eq:pre_filter}.
Each realization is then sampled with sampling rate $\T = \T_{\mathrm{US}},\ldots,\T_{\mathrm{Shannon}}$ in $100$ steps and we run Algorithm~\ref{algo:UL_lambda} with
\begin{equation*}
N \in \biggl\{j \, \Bigl\lceil\frac{\log(\lambda)}{\log(0.5\T_{\mathrm{US}}\Omega\e)}\Bigr\rceil \biggm| j=1,2,3\biggr\}.
\end{equation*}

For each combination of $\T$ and $N$ we compute the {\em success rate} which is defined as the fraction of trials in which the algorithm reconstructs the samples $g(k\T)$ from its modulo samples~$\Mod_\lambda(g(k\T))$.
The results are summarized in Figure~\ref{fig:US_lambda_success}.
We observe that indeed the recovery is possible even when $\T \geq \T_{\mathrm{US}}$.
Moreover, a smaller threshold $\lambda$ requires a slightly faster sampling rate  below which the algorithm always succeeds, whereas a larger bandwidth $\Omega$ allows for a slower sampling rate.
On top of that, higher order differences also allow for smaller oversampling factors up to approximately $\T = 0.3 \cdot \T_{\mathrm{Shannon}}$, but there seems to be a minimum sampling rate around $\T = 0.4 \cdot \T_{\mathrm{Shannon}}$ above which the algorithm always fails.


\section{Conclusions and Looking Ahead}

The key finding of this paper is a single-shot solution to the problem of high dynamic range (HDR) tomography that is based on a co-design approach of novel hardware and mathematically guaranteed recovery algorithms. To this end, we rigorously introduce the Modulo Radon Transform and study its properties and inversion. On the hardware front, we conceptualize a detector that records the modulo of the Radon Transform projections. The implication is that HDR projections, that may be arbitrarily larger than the detector's recordable range, are folded back and registered as low dynamic range measurements. This avoids saturation or clipping problems. On the algorithmic front, we study different scenarios for recovery of the image or object from modulo-folded projections. Starting with band-limited functions, we study injectivity conditions for identifiability of the modulo projections. This also results in a recovery guarantee specified in terms of a sampling density criterion. Thereon, taking a step closer to practice, we study the case of approximately compactly supported functions. To analyze the recovery in this case, we introduce the $\lambda$-exceedance property. The upshot of this approach is a significant reduction in additionally needed samples.

Beyond the theoretical advances, as a proof-of-concept, we also demonstrate that our approach can be brought closer to practice. To do so, we use an open source data set comprising of Radon Transform projections that are folded in hardware using our custom designed modulo sampling prototype \cite{Bhandari2021J}. This sequential hardware setup gives a realistic sense of uncertainties and noise that one may encounter in practice. While our theory does not yet provide recovery guarantees for specific noise models, the inherent stability of our algorithms allows us to reconstruct HDR data with high fidelity. In this sense, our all-hardware demo establishes the practicability of our approach.  

\subsection*{New Avenues in Interdisciplinary Areas}
Since the Modulo Radon Transform harnesses a joint design between hardware and algorithms, it is natural that our work has interesting implications in interdisciplinary areas. We enlist a few interesting research directions. 

\subsubsection*{Theoretical Considerations} 
\begin{itemize}
  \item Noise and Perturbations. Our current work is based on the inversion of finite-difference operators, and the stability arises from the fact that the modulo operation defines a natural discrete grid. At the same time, for small modulo thresholds, this advantage becomes less pronounced, and additional measures to increase the noise robustness should be considered. In particular, since the end result of our imaging pipeline is a digital signal, the effects of bounded noise, arising from quantization, remains to be explored. 
  \item Implicit Trade-offs.  Since the HDR capability of our current algorithm is intricately tied to the order of difference operators and the sampling rate, a natural next step is to understand the implicit trade-offs in the recovery procedure. We remind the reader that our simulations show that the sampling bound prescribed in our work can be tightened. This remains an interesting research pursuit because, for fixed difference order, a tighter bound will allow us to operate with higher dynamic signals.
  \item Sample Sizes. We have observed that the $\lambda$-exceedance property results in lower sample sizes. We believe that the exact relation between bandwidth and $\lambda$-exceedance will pave a path for designing practical systems, giving precise guidelines about realistic sample sizes.
\end{itemize}

\subsubsection*{Practical Considerations} 

\begin{itemize}
  \item Signal Spaces. In addition to the band-limited model arising as a consequence of the inherent smoothing, we hope to extend our results to refined models that can additionally capture non-band-limited features that naturally arise in tomographic images. 

  \item Algorithms. For practical implementation, due to huge data sizes, fast and efficient recovery algorithms are highly desirable. Our current algorithms are sequential and implement an ``unfold then reconstruct'' approach. This sequential aspect of reconstruction may slow down the recovery procedure in practice. A natural next goal is to develop a fast and robust, single-shot, recovery approach.
  \item Imaging Geometry. As an inaugural example, in this work, we considered parallel beam imaging geometry. To increase the utility of our approach, we plan to extend our work to fan-beam and cone-beam geometries. 
\end{itemize}


\bibliographystyle{siamplain}

\end{document}